\documentclass{article}
\usepackage[utf8]{inputenc}
\usepackage{fullpage}

\usepackage{amsmath,amssymb,pifont}
\usepackage[numbers,sort]{natbib}
\usepackage{multicol}
\usepackage{amstext}
\usepackage{amsthm}
\usepackage{multirow}
\usepackage{booktabs}
\usepackage{adjustbox}
\usepackage[skip=0pt,format=hang,singlelinecheck=false]{subcaption}
\usepackage{lipsum}
\usepackage[shortlabels]{enumitem}
\usepackage{cancel}
\usepackage{wrapfig}
\usepackage{array}
\usepackage{siunitx}
\usepackage{csvsimple}
\usepackage[multidot]{grffile}
\usepackage{bbm}
\usepackage{dblfloatfix}
\usepackage[colorlinks = true,
            linkcolor = blue,
            urlcolor  = blue,
            citecolor = blue,
            anchorcolor = blue]{hyperref}
\usepackage{makecell}
\usepackage{bbm, dsfont}
\usepackage{mathtools}
\usepackage{comment}

\usepackage{geometry}
\usepackage{enumitem}

\usepackage{minitoc} %

\usepackage[multiple]{footmisc}
\usepackage{mathrsfs}
\usepackage{todonotes}
\usepackage{tikz}
\usepackage[capitalise]{cleveref}
\usepackage{nicematrix}

\usepackage[noend]{algpseudocode}
\usepackage{algorithm,algorithmicx}
\algrenewcommand\algorithmicrequire{\textbf{Input:}}
\algrenewcommand\algorithmicensure{\textbf{Return}}
\algnewcommand{\LineComment}[1]{\Statex \textcolor{blue}{$\blacktriangleright$ \textit{#1}}}
\algrenewcommand{\algorithmiccomment}[1]{\hfill\textcolor{blue}{$\blacktriangleright$ \textit{#1}}}  %
\usepackage{xfrac}

\usepackage{graphicx} %
\usepackage{color}

\usepackage{array}
\usepackage{amssymb}
\usepackage{amsmath}
\usepackage{xspace}
\usepackage{fancyhdr}
\usepackage{comment}
\usepackage{bm}
\usepackage{bm}
\newcommand{\op}[1]{\operatorname{#1}}

\makeatletter
\def\renewtheorem#1{%
  \expandafter\let\csname#1\endcsname\relax
  \expandafter\let\csname c@#1\endcsname\relax
  \gdef\renewtheorem@envname{#1}
  \renewtheorem@secpar
}
\def\renewtheorem@secpar{\@ifnextchar[{\renewtheorem@numberedlike}{\renewtheorem@nonumberedlike}}
\def\renewtheorem@numberedlike[#1]#2{\newtheorem{\renewtheorem@envname}[#1]{#2}}
\def\renewtheorem@nonumberedlike#1{  
\def\renewtheorem@caption{#1}
\edef\renewtheorem@nowithin{\noexpand\newtheorem{\renewtheorem@envname}{\renewtheorem@caption}}
\renewtheorem@thirdpar
}
\def\renewtheorem@thirdpar{\@ifnextchar[{\renewtheorem@within}{\renewtheorem@nowithin}}
\def\renewtheorem@within[#1]{\renewtheorem@nowithin[#1]}
\makeatother

\definecolor{darkgreen}{rgb}{0,0.4,0.0}
\definecolor{C0}{HTML}{1F77B4}
\definecolor{C1}{HTML}{ff7f0e}
\definecolor{C2}{HTML}{2ca02c}
\definecolor{C3}{HTML}{d62728}
\definecolor{C4}{HTML}{9467bd}

\newtheorem{theorem}{Theorem}
\newtheorem{lemma}[theorem]{Lemma}

\newtheorem{property}[theorem]{Property}
\newtheorem{proposition}[theorem]{Proposition}
\newtheorem{corollary}[theorem]{Corollary}
\newtheorem{remark}[theorem]{Remark}

\newtheorem*{example*}{Example}

 \crefname{property}{Property}{Properties}

\newcommand{\prefix}{\bfA_{\mathsf{pre}}}
\newcommand{\strategy}{\bfC}

\newcommand{\BLT}{\mathsf{BLT}}

\newcommand{\genf}{f}
\newcommand{\hatgenf}{\hat{\genf}}

\newcommand{\iind}{_{i=1}^d}
\newcommand{\iinD}{_{i=1}^D}

\newcommand{\nd}{\ensuremath{n}} %

\newcommand{\mdim}{m}  %
\newcommand{\inv}{^{-1}}

\newcommand{\sens}{\ensuremath{\mathsf{sens}}}

\newcommand{\bfA}{{\bm{A}}}
\newcommand{\bfB}{{\bm{B}}}
\newcommand{\bfC}{{\bm{C}}}

\newcommand{\bfG}{{\bm{G}}}

\newcommand{\bfI}{{\bm{I}}}

\newcommand{\bfM}{{\bm{M}}}

\newcommand{\bfW}{{\bm{W}}}

\newcommand{\bfZ}{{\bm{Z}}}

\newcommand{\bfg}{{\bm{g}}}

\newcommand{\bfu}{{\bm{u}}}
\newcommand{\bfv}{{\bm{v}}}

\newcommand{\bftheta}{\ensuremath{\bm{\theta}}}

\newcommand{\bfalpha}{\ensuremath{\bm{\alpha}}}

\newcommand{\bflambda}{\ensuremath{\bm{\lambda}}}

\newcommand{\norm}[1]{\left\|{#1}\right\|}

\newcommand{\diag}{\ensuremath{\mathrm{diag}}}

\newcommand{\lfrob}[1]{\left\|#1\right\|_{\op{F}}}

\newcommand{\rownorm}[1]{\left\| #1 \right\|_{\mathsf{row}}}
\newcommand{\colnorm}[1]{\left\| #1 \right\|_{\mathsf{col}}}

\newcommand{\T}{^{\top}}  %

\newcommand{\R}{\mathbb{R}}
\newcommand{\C}{\mathbb{C}}

\renewcommand{\epsilon}{\varepsilon}
\DeclarePairedDelimiterX{\inp}[2]{\langle}{\rangle}{#1, #2}  %

\newcommand{\mech}{\mathcal{M}}

\newcommand{\normal}{\mathcal{N}}
\newcommand{\normaldim}[3]{\normal_{#1}\big(#2, #3\big)}

\newcommand{\normalnm}[2]{\normaldim{n\!\times\!m}{#1}{#2}}

\newcommand{\workload}{\bfA}

\newcommand{\zeros}{{\bm{0}}}

\newcommand{\decay}{\lambda}
\newcommand{\decayv}{\bflambda}
\newcommand{\scale}{\alpha}
\newcommand{\scalev}{\bfalpha}

\newcommand{\polyDiff}{r}  %
\newcommand{\polyDenom}{p}  %
\newcommand{\polyInvDenom}{q}  %
\title{An Inversion Theorem for Buffered Linear Toeplitz (BLT) Matrices and Applications to Streaming Differential Privacy}
\author{H. Brendan McMahan$^1$
\and
Krishna Pillutla$^2$
}
\date{
\small
$^1$Google Research $\quad$
$^2$Wadhwani School of Data Science \& AI, IIT Madras
}

\begin{document}

\maketitle
\begin{abstract}
Buffered Linear Toeplitz (BLT) matrices are a family of parameterized lower-triangular matrices that play an important role in streaming differential privacy with correlated noise. Our main result is a BLT inversion theorem: the inverse of a BLT matrix is itself a BLT matrix with different parameters.  
We also present an efficient and differentiable $O(d^3)$ algorithm to compute the parameters of the inverse BLT matrix, where $d$ is the degree of the original BLT (typically $d < 10$). Our characterization enables direct optimization of BLT parameters for privacy mechanisms through automatic differentiation.
\end{abstract}

\section{Introduction}
\label{sec:intro}

We consider the inverses of a family of parameterized lower-triangular matrices known as Buffered Linear Toeplitz (BLT) matrices~\cite{dvijotham2024efficient}. Given a scale parameter $\scalev = (\scale_1, \ldots, \scale_{d}) \in \R^d$ and a decay parameter $\decayv = (\decay_1, \ldots, \decay_{d}) \in \R^d$, the $n \times n$ BLT matrix is defined as
\begin{align} \label{eq:blt-param}
    \BLT_n(\scalev, \decayv)
    := \begin{pmatrix}
    1 & 0 & 0 & 0 & \cdots \\
    \sum\iind \scale_i & 1 & 0 & 0 & \cdots  \\
    \sum\iind \scale_i \decay_i  & \sum\iind \scale_i & 1 & 0 & \cdots \\ 
    \sum\iind \scale_i \decay_i^2  &  \sum\iind \scale_i \decay_i   & \sum\iind \scale_i & 1 & \ddots \\ 
    \vdots  & \vdots & \vdots & \vdots & \ddots
    \end{pmatrix}_{n \times n} \,.
\end{align}
The matrix $\strategy = \BLT_n(\scalev, \decayv)$ is lower triangular and Toeplitz (i.e., it has equal entries along each diagonal from the top-left to the bottom-right), with ones along the principal diagonal:\footnote{
    We denote the $(j,k)$\textsuperscript{th} entry of the matrix $\strategy \in \R^{n \times n}$ by $\strategy[j, k]$; \cref{tab:notation} in \Cref{sec:notation} provides a complete notation summary.
}
\begin{equation}\label{eq:coefs}
    \strategy[j, k] = \begin{cases}
        0, & \text{if } j < k \,, \\
        1, & \text{if } j = k \,, \\
        \sum\iind \scale_i \decay_i^{j-k-1}, & \text{if } j > k\,.
    \end{cases}
\end{equation}
Such parameterized matrices (and their inverses) are central to streaming differential privacy with correlated noise, achieving near-optimal tradeoffs between privacy, utility, and computation cost; we describe this in detail in \Cref{sec:dp-bg}.

The main result of this note is that the BLT family of matrices is closed under inversion. In particular, the inverse of a degree-$d$ BLT matrix $\BLT_n(\scalev, \decayv)$ (which always exists) is also a BLT matrix of the same order $d$: 
\[
    \BLT_n(\scalev, \decayv)^{-1} = \BLT_n(\hat\scalev, \hat\decayv) \quad \text{for all } n > 0
\]
for all integers $n > 0$ for \emph{unique} parameters $\hat\scalev, \hat\decayv \in \R^d$.
We also give an equivalence between representing a (BLT, inverse BLT) system $\BLT_n(\scalev, \decayv)^{-1} = \BLT_n(\hat \scalev, \hat \decayv)$ using (a) both the parameters $\scalev, \decayv$ of the first BLT, and (b) both the decay parameters $\decayv, \hat \decayv$.
Finally, we give a differentiable algorithm to compute the BLT inverse in $O(d^3)$ time for any size $n$.

Next, we provide some background on streaming differential privacy in \Cref{sec:dp-bg}. BLT matrices and their inverses play an important role in this setting. We give a full formal statement of our BLT inversion theorem in \Cref{sec:main-results}.
We give the key ideas behind the proofs in \Cref{sec:proof-main} with full proof details in \Cref{sec:proof}.

\section{Background}
\label{sec:dp-bg}

Let $\bfG \in \R^{\nd \times \mdim}$ be a sequence of $\nd$ vectors in $\R^\mdim$ stacked row-wise into a matrix. Each vector $\bfg_t$ (i.e. $t$\textsuperscript{th} row of $\bfG$) is assumed to satisfy $\norm{\bfg_t}_2 \le \zeta$ for some constant $\zeta > 0$.
We aim to estimate (in a differentially private manner) a sequence of (known) linear combinations of these vectors,  represented as the rows  of $\workload \bfG \in \R^{n\times m}$; here, $\workload \in \R^{\nd \times \nd}$ is known as the \emph{workload matrix}.

This setup captures diverse problems such as continual counting and stochastic optimization under differential privacy. In the latter case, $\bfg_t$ is an unbiased estimator of the loss gradient evaluated at the current model parameters $\bftheta_t$.
The workload matrix captures the optimization algorithm: stochastic gradient descent (SGD) with a constant learning rate $\eta$ corresponds to the prefix sum workload $\prefix$, which is the lower triangular matrix with all ones. This is because each iterate $\bftheta_t = \bftheta_0 - \eta \sum_{\tau < t} \bfg_\tau$ of SGD relies on estimating the prefix sums $\sum_{\tau < t} \bfg_\tau$, which are the rows of $\prefix \bfG$. Other first-order optimizers such as SGD with momentum correspond to different workloads.

The matrix mechanism for differential privacy~\cite{li2015matrix,nikolov2016geometry}, known also as DP-FTRL~\cite{kairouz2021practical,denisov2022improved} in the learning setting, injects correlated noise to release private estimates of $\workload\bfG$. Given a factorization $\workload = \bfB \strategy{}$ with $\strategy{}$ invertible, this correlated noise mechanism is defined as
\begin{align} \label{eq:mech}
    \mech(\bfG) = \bfB(\strategy{}\bfG + \bfZ) = \workload{}(\bfG + \strategy^{-1} \bfZ)\,,
\end{align}
where $\bfZ \in \R^{\nd \times \mdim}$ is component-wise i.i.d. Gaussian noise. We scale\footnote{
    The notation $\bfZ \sim \normalnm{0}{\sigma^2}$ denotes a random matrix  $\bfZ \in \R^{n \times m}$ whose entries are i.i.d. $\normal(0, \sigma^2)$.
}  $\bfZ \sim \normalnm{0} {\sens(\strategy{})^2 \sigma^2}$,
where $\sens(\strategy{})$ is the $\ell_2$-sensitivity of the operation $\bfG \mapsto \strategy{}\bfG$, while $\sigma$ is a noise multiplier depending only on the desired privacy level; e.g. we take $\sigma^2 = 1 / (2\rho)$ for a $\rho$-zero-concentrated DP guarantee~\cite{bun16concentrated}. Note that multiplication by $\bfB$ is simply a post-processing step that does not affect the privacy guarantee. 
The sensitivity $\sens(\strategy{})$ depends on how adjacent $\bfG, \bfG'$ are allowed to differ. In the learning setting, if a data item can appear only once in training, then $\sens(\strategy) = \colnorm{\strategy}$ is the maximum column norm of the matrix $\strategy{}$. Different expressions exist when each data item can participate more than once~\cite[see e.g.][Eq. (2)]{choquette2023amplified}. 

In general, we aim to find the factorization $\workload=\bfB\strategy$ to minimize the worst-case expected (squared) $\ell_2$ norm across all rows of $\mech(\bfG) - \workload \bfG = \bfB \bfZ = \workload^{-1} \strategy{} \bfZ$. This can be evaluated (assuming the norm constant $\zeta=1$ w.l.o.g.) as the (square of the) \emph{max loss}
\begin{align} \label{eq:maxloss}
    L(\bfC) := \sens(\strategy) \cdot \rownorm{\workload^{-1}\strategy{}}\,,
\end{align}
where $\rownorm{\bfB}$ denotes the maximum row norm of the matrix $\bfB$. Sometimes, we may choose the Frobenius norm $\lfrob{\cdot}$ instead of $\rownorm{\cdot}$ to compute the average expected (squared) $\ell_2$ norm across rows of $\workload\bfG$ instead of the worst-case.

\paragraph{BLT Mechanism}
A major focus of prior research has been to improve the privacy-utility-compute tradeoffs of the mechanism \eqref{eq:mech} in theory and practice; see \cite{fichtenberger2023constant,choquette2023amplified,mckenna2024scaling,henzinger2025improved} and the references therein. 
In particular, the computation cost of computing $(\strategy\inv \bfZ)[t, :]$ in each iteration $t$ dominates the running time of the algorithm in the learning setting.
The \emph{BLT mechanism}~\cite{dvijotham2024efficient} achieves state-of-the-art tradeoffs. In general, it restricts the $\strategy{}$ matrix to be Toeplitz and parameterizes its first column $c_1, c_2, \ldots$ as $c_t = \bfu\T \bfW^{t-1} \bfv$ using a matrix $\bfW$ and two vectors $\bfu, \bfv$.
We focus on the \emph{diagonal BLT} formulation described in \cref{eq:blt-param}, which corresponds to diagonal $\bfW$; this formulation has been preferred in empirical studies for being more computationally efficient without sacrificing utility~\cite{dvijotham2024efficient,mcmahan2024hassle}.

A key advantage of the BLT mechanism is that the rows of the correlated noise $\strategy{}^{-1}\bfZ$ in \cref{eq:mech} can be generated in a streaming fashion with $O(d \mdim)$ time and space complexity~\cite[cf.][Alg. 2,3]{mcmahan2024hassle}; notably, this is independent of the iteration counter.
Together with additive utility guarantee in the streaming setting where $\sens(\strategy) = \colnorm{\strategy}$, this leads to near-optimal privacy-utility-compute tradeoffs with the prefix sum workload $\workload=\prefix$. In particular, for any size $n > 0$ and error term $\delta > 0$, there exist some parameters $\scalev, \decayv \in \R^d$ for $d = O(\log^2(n/\delta))$ that give an \emph{additive} approximation of the optimal max error:
\begin{align} \label{eq:blt-approx}
    L\big(\BLT_n(\scalev, \decayv)\big)
    \le \min\{ L(\strategy) \,: \strategy\in \R^{n \times n} \text{ is lower-triangular \& Toeplitz } \} + \delta \,.
\end{align}

In this work, we show that the inverse of a diagonal BLT of the form of \cref{eq:blt-param} is another diagonal BLT; we give a precise statement in \Cref{sec:main-results}. We also describe how to find the parameters $\hat \decayv, \hat\scalev$ of the inverse BLT in a differentiable manner so that max loss \eqref{eq:maxloss} can be optimized (as a function of the BLT parameters $\scalev, \decayv$) using automatic differentiation.

\paragraph{Parameter Restrictions}
We restrict ourselves to the BLT decay parameter $\decayv \in (0, 1)^d$, and our main result (\Cref{thm:main}) further focuses on the case where each $\scale_i > 0$ and $\sum\iind \scale_i < 1$.  These restrictions are not strictly necessary, in that \cref{eq:blt-param} is a well-defined (and invertible) matrix for any parameters $\scalev, \decayv \in \R^d$.
Why these restrictions? In short, we believe they identify the most practically important subclass of BLTs where the goal is to approximate the optimal Toeplitz matrix (see the right side of \cref{eq:blt-approx}).
This allows sharper and simpler theoretical characterizations and numerically stable mechanisms. 
The restriction $\scalev > 0$ (or $\scalev < 0$) is beneficial when optimizing BLTs, and ensure our subclass is closed under matrix inversion. 
For example, the previous works \citep{dvijotham2024efficient,mcmahan2024hassle} restrict the search over $\scalev$ to over strictly positive entries by imposing log-barrier functions in the optimization.

 Do these restrictions lead to sub-optimal mechanisms? For the problem of correlated noise DP mechanisms for single-participation, strong empirical evidence from prior work \citep{dvijotham2024efficient,mcmahan2024hassle} shows that this is not the case. In particular,
 \citet{dvijotham2024efficient} showed BLTs that satisfy these restrictions can for all practical purposes perfectly match the optimal Toeplitz matrix (i.e. the Toeplitz matrix $\strategy$ minimizing $L(\strategy)$ optimizing, as in the right size of \cref{eq:blt-approx}), and \citet{mcmahan2024hassle} showed strong performance for a common multiple-participation setting. Nevertheless, it is possible that for some applications the additional expressive power of allowing some $\decay_i < 0$ could make it worth investigating this case further.

\paragraph{Notation Summary}
We use the shorthand $[d] := \{1, \ldots, d\}$.
Vectors are denoted by boldfaced lower-case (Greek or Latin) letters (e.g. $\decayv$, $\bfu$) while matrices are denoted by boldfaced upper-case letters (e.g. $\strategy$ or $\bfM$); both are 1-indexed.
Often, we will denote the first columns of the lower triangular and Toeplitz matrices $\strategy$ and $\strategy^{-1}$ with shorthand $c_t = \strategy[t, 1]$ and $\hat c_t = (\strategy\inv)[t, 1]$ respectively.
We give a detailed summary of the notation in \Cref{tab:notation} of \Cref{sec:notation}.

\section{Main Results}
\label{sec:main-results}

Our main result is a BLT inversion theorem: the inverse of a BLT matrix is also a unique BLT. We give some properties of the inverse BLT parameters. All proofs are given in \Cref{sec:proof-main,sec:proof}.
We say a vector of parameters $\decayv$ is \textbf{distinct} if it holds that
\[ 
 \decay_i \neq \decay_j 
 \qquad 
 \text{for all $i, j \in [d]$ such that $i \neq j$}.
\]

\begin{theorem}
\label{thm:main}
    The matrix $\BLT_n(\scalev, \decayv)$ is invertible for any integer $n > 0$ for any parameters $\scalev \in \R^d$ and $\decayv \in \R^d$ for all integers $n > 0$ and $d > 0$.
    In addition, if the scale parameters are positive ($\scale_i > 0$) and satisfy $\sum\iind \scale_i < 1$, and the decay parameters $\decayv \in (0, 1)^d$ are distinct, then there exist parameters $\hat\scalev, \hat\decayv \in \R^d$  with $\hat \decayv$ distinct such that $\BLT_n(\scalev, \decayv)^{-1} = \BLT_n(\hat\scalev, \hat\decayv)$ for all integers $n > 0$.
    Further, the scale parameters of the inverse are negative (i.e. $\hat \scale_i < 0$ for all $i$), and the decay parameters of the inverse satisfy have the following:
    \begin{enumerate}[label=(\alph*)]
        \item \label{part:thm:lt-1}
        If $\sum\iind \scale_i / \decay_i < 1$, then $\hat\decay_i \in (0, 1)$ for each $i\in [d]$. 
        \item \label{part:thm:gt-1}
        If $\sum\iind \scale_i / \decay_i > 1$, then there exists an integer $j \in [d]$ such that $\hat\decay_j \in (-1, 0)$ and $\hat\decay_i \in (0, 1)$ for all $i \in [d], i \neq j$.
        \item \label{part:thm:eq-1}
        Finally, if $\sum\iind \scale_i / \decay_i = 1$, then there exists an integer $j \in [d]$ such that $\hat\decay_j= 0$ and $\hat\decay_i \in (0, 1)$ for all $i \in [d], i \neq j$.
    \end{enumerate}
    Furthermore, these inverse parameters $\hat \decayv, \hat \scalev$ are unique (up to permutations of indices).
\end{theorem}

Following \cref{eq:coefs},  $\strategy[2, 1] = \sum\iind \scale_i$, and so the assumption  $\sum\iind \scale_i < 1$ is useful because we typically want the first column of the $\strategy$ matrix to be decreasing in the context of streaming differential privacy \cite{dvijotham2024efficient,mcmahan2024hassle}.
Furthermore, we empirically observe that BLT parameters optimized for the max loss tend to satisfy $\hat \decay_i \in (0, 1)$, paralleling the assumption that $\decay_i \in (0, 1)$ for the original BLT; this is also true for the theoretical construction of \citet{dvijotham2024efficient}.
This corresponds to the regime of $\sum\iind \scale_i / \decay_i < 1$ as per \Cref{thm:main}\ref{part:thm:lt-1}.
BLTs satisfying the condition of \Cref{thm:main}\ref{part:thm:eq-1} are degenerate in the sense that one of the decay parameters is exactly zero. We give an example where this holds.
\begin{example*} 
    Consider a BLT of degree $d=2$ with parameters 
    $\scalev=({2}/{5}, {1}/{5})$ and $\decayv=({4}/{5}, {2}/{5})$
    so that $\sum_i \scale_i / \decay_i = 1/2 + 1/2 = 1$. Then, we have that its inverse is a BLT given by the parameters\footnote{
        This can be verified, for instance, using the upcoming \Cref{lem:generating-fn}.
    } 
    $\hat \scalev = (-1/15, -8/15)$ and $\hat \decayv = ({3}/{5}, 0)$. In particular, notice that we have a decay parameter of exactly $0$ in the inverse, and hence $\hat\scale_2 = -8/15$ only influences the second Toeplitz coefficient $(\strategy^{-1})[j+1, j] = \sum\iind \hat\scale_i$.
\end{example*}

For practical applications, we usually optimize for the BLT parameters, and hence will not reach the measure zero set of parameters with $\sum_i \scale_i / \decay_i = 1$ when the optimizer is initialized randomly. Further, a small numerical error in the BLT parameters is enough to make this degeneracy vanish.

\paragraph{Comparison to Previous Literature}
We make two remarks where \Cref{thm:main} significantly simplifies and extends prior work.
First, \citet[Lemma 5.2]{dvijotham2024efficient} showed (using our notation) that given an appropriate pair of decay parameters $(\decayv, \hat\decayv)$, there exist scale parameters $(\scalev, \hat\scalev)$ such that $\BLT(\scalev, \decayv)\inv = \BLT(\hat\scalev, \hat\decayv)$. 
We strengthen this result with the help of \Cref{thm:main} to show that a (BLT, inverse BLT) system can be parameterized in two equivalent ways. We also give significantly simplified expressions compared to \citep[Lemma 5.2]{dvijotham2024efficient}.
\begin{theorem} \label{thm:calc-output-scale}
    Let $\decayv, \hat \decayv \in \R^d$ be distinct non-zero vectors (i.e., $\decay_i \neq \decay_j$ and $\hat \decay_i \neq \hat \decay_j$ for all $i \neq j$) that also satisfy $\decay_i \neq \hat\decay_j$ for all $i, j \in [d]$.
    Then, there exist unique scale parameters $\scalev, \hat\scalev \in \R^d$ that achieve $\BLT_n(\scalev, \decayv) = \BLT_n(\hat\scalev, \hat \decayv)^{-1}$ for all $n > 0$, given by:
    \begin{align} \label{eq:calc-output-scale}
        \scale_i &=  \frac{\prod_{j=1}^d \decay_i - \hat \decay_j}{\prod_{j \neq i} \decay_i - \decay_j}\,,
        \quad\text{and} \quad
        \hat\scale_i = \frac{\prod_{j=1}^d \hat\decay_i - \decay_j}{\prod_{j \neq i} \hat\decay_i - \hat\decay_j} \,.
    \end{align}
    Furthermore, the following two parameterizations describe the same class of (BLT, inverse BLT) systems satisfying $\BLT(\scalev, \decayv)^{-1} = \BLT(\hat\scalev, \hat\decayv)$:
    \begin{enumerate}[label=(\alph*), nosep]
    \item \label{item:pos-param}
    positive scale parameters $\scalev \in \R^d_{++}$ and distinct decay parameters $\decayv \in (0, 1)^d$ of the BLT that satisfy $\sum\iind \scale_i / \decay_i < 1$;
    \item \label{item:decay-param}
    a pair of decay parameters $\decayv, \hat \decayv$ that satisfy the strict interlacing condition
    \[
    1 > \decay_1 > \hat \decay_1 > \decay_2 > \hat \decay_2 > \cdots > \hat \decay_{d-1} > \decay_d > \hat \decay_d > 0 .
    \]
    \end{enumerate}
    Given either parameterization, the system $\BLT(\scalev, \decayv)^{-1} = \BLT(\hat\scalev, \hat\decayv)$ is uniquely determined.
\end{theorem}

Next, we turn to BLT inversion theorems.
\citet[Proposition 5.6]{dvijotham2024efficient} show that the inverse of $\BLT_n(\scalev, \decayv)$ is a lower-triangular and Toeplitz matrix whose first column $\hat c_1, \hat c_2, \ldots$ is given by 
\begin{align} \label{eq:old-inversion-thm}
\hat c_t = 
\begin{cases}
\bfu\T \bfv + \kappa \,, & \text{ if } t=1\,, \\
\bfu\T \bfW^{t-1} \bfv \, &\text{ if } t > 1 \,.
\end{cases}
\end{align}
for vectors $\bfu, \bfv \in \R^{d}$, a matrix $\bfW \in \R^{d \times d}$ and a scalar $\kappa \in R$.\footnote{
This representation is not unique, as $\tilde\bfu = \bfM \bfu$, $\tilde\bfW = \bfM \bfW \bfM\T$, $\tilde \bfv = \bfM\bfv$ satisfies $\bfu\T \bfW^\tau \bfv = \tilde \bfu\T \tilde \bfW^\tau \tilde \bfv$ for all $\tau \ge 0$ for any orthonormal matrix $\bfM$.
}
\Cref{thm:main} implies that we can instead take
\[ %
    \bfW = 
    \begin{pmatrix}
        \hat\decay_1 \\
        & \ddots \\
        & & \hat\decay_d
    \end{pmatrix} \in \R^{d \times d}, \quad 
    \bfu = \begin{pmatrix}
    1 \\  \vdots \\ 1
    \end{pmatrix} \in \R^{d}, \quad
    \bfv = \begin{pmatrix}
        \hat\scale_i / \hat\decay_i \\
        \vdots \\
        \hat\scale_d / \hat\decay_d
    \end{pmatrix} \in \R^{d} \,, 
    \quad 
    \kappa = 1 - \sum\iind \frac{\scale_i}{\decay_i}
    \,.
\]
We can verify by direct computation that this produces $\hat c_1 = 1$ and $\hat c_t = \sum\iind \hat\scale_i \hat\decay_i^{t-2}$ for $t \ge 2$, as desired.
In other words, \Cref{thm:main} shows the existence of a $\bfW$ of rank $d$ and all real eigenvalues satisfying \cref{eq:old-inversion-thm}. 

The representation in \Cref{thm:main} is more convenient from a computational perspective, as operations on diagonal $\bfW$ can be implemented more efficiently.
Of course, given any $\bfu, \bfv, \bfW$ that satisfy $\hat c_t = \bfu\T \bfW^{t-1} \bfv$ (e.g. by the approach of \cite[Proposition 5.6]{dvijotham2024efficient}), we can then find the inverse BLT decay parameter $\hat\decayv$ by diagonalizing
$\bfW = \bfM \, \diag(\hat\decayv) \, \bfM^{-1}$, \emph{assuming} it is possible, and scale parameter $\hat\scalev = (\bfM\T \bfu) \odot (\bfM^{-1} \bfv)$, where $\odot$ denotes component-wise multiplication of vectors.
By showing that $\bfW$ has all real and unique eigenvalues, \Cref{thm:main} establishes that this matrix is diagonalizable.

\paragraph{Algorithms for BLT Inversion}
We give an algorithm to directly find the parameters $\hat\scalev$ and $\hat\decayv$ of the inverse BLT. Specifically, the decay parameter $\hat\decayv$ is obtained from the finding the roots of a degree-$d$ polynomial $\polyInvDenom$ (whose roots are guaranteed to all be real). 
While these roots can be analytically obtained for degrees $d = 4$ or lower, analytical expressions for the roots of general polynomials of degree $d \ge 5$ are impossible. (This is also known as the Abel–Ruffini theorem, see e.g. \cite{ayoub1980ruffini,ramond2022abel}.)

Instead, we use numerical polynomial root-finding procedures to find all the roots of the polynomial $\polyInvDenom$. The typical approach proceeds by constructing a non-symmetric matrix, known as the companion matrix, and then finding its eigenvalues in $O(d^3)$ time (see e.g. \cite{roots} or the upcoming \Cref{sec:algo})---these eigenvalues are exactly the roots of the polynomial $\polyInvDenom$. 
Importantly, all these operations are supported by typical automatic differentiation frameworks, including JAX and PyTorch. 
This allows us to parameterize the optimization of a BLT and its inverse as $(\scalev, \decayv)$, rather than as $(\decayv, \hat\decayv)$ as in \cite{dvijotham2024efficient}.
See \Cref{sec:algo} for details.

\section{Technical Tools and Proof Outline}
\label{sec:proof-main}

The proofs of \Cref{thm:main,thm:calc-output-scale} rely on a deep connection between Toeplitz matrices and ordinary generating functions.

The \textbf{ordinary generating function} of a sequence $(c_t)_{t=1}^\infty$ is the formal power series:\footnote{
    The variable $x$ in a formal power series should be interpreted as a formal symbol rather than a numerical value. Specifically, we neglect any concerns related to convergence.
}
\[
    f_c(x) := \sum_{t=0}^\infty c_{t+1} x^t \,.
\]
This is closely related to the $Z$-transform, which can be obtained by the symbolic substitution $z = 1/x$.
The sequence $(c_t)_{t=1}^\infty$ can be obtained from the Maclaurin expansion of the generating function:
\[
    f_c(x) = \sum_{t=0}^\infty \frac{f_c^{(t)}(0)}{t!} \, x^t  = \sum_{t=0}^\infty c_{t+1} x^t 
    \quad \iff \quad 
    c_{t+1} = \frac{f_c^{(t)}(0)}{t!} \,,
\]
where $f_c^{(t)}$ denotes the $t$\textsuperscript{th} derivative of the function $f_c$ (assuming it exists).

The key relationship between lower-triangular Toeplitz matrices and generating functions is that the product of two Toeplitz matrices (i.e., the convolution of their first columns) is equivalent to the product of the generating functions of their respective coefficients. (This is analogous to the fact that the convolution of two sequences can be obtained from the product of their $Z$-transforms.)
\begin{lemma} \label{lem:generating-fn}
    For any sequences $(a_t)_{t=1}^\infty, (b_t)_{t=1}^\infty, (c_t)_{t=1}^\infty$ that take values in a field $K$(e.g. the field $\R$ of reals or $\C$ of complex numbers), the following are equivalent:
    \begin{enumerate}[label=(\roman*),nosep]
        \item The respective ordinary generating functions $f_a, f_b$, and $f_c$ of sequences $(a_t)_{t=1}^\infty, (b_t)_{t=1}^\infty$, and $(c_t)_{t=1}^\infty$ satisfy $f_a(x) = f_b(x) f_c(x)$.
        \item For any integer $n > 0$, the $n \times n$ lower triangular Toeplitz matrices $\bfM_a, \bfM_b, \bfM_c$ with respective first columns given by sequences $(a_t)_{t=1}^{n}, (b_t)_{t=1}^{n}, (c_t)_{t=1}^{n}$ satisfy $\bfM_a = \bfM_b \bfM_c$.
    \end{enumerate}
\end{lemma}

We are most interested in real sequences. In particular, \Cref{lem:generating-fn} tells us that we can calculate the inverse $\hat \strategy = \strategy^{-1}$ of any lower-triangular Toeplitz matrix $\strategy$ whose first column is obtained as a prefix of the sequence $(c_t)_{t=0}^\infty$
by the following steps: 
\begin{itemize}
    \item Compute its generating function $f_c(x)$;
    \item Calculate the reciprocal $\hat f_c(x) = 1/f_c(x)$;
    \item Calculate its Maclaurin series $\hat c_{t+1} = {\hat f_c^{(t)}(0)} / (t!)\,\, $ for $t \ge 0$;
    \item Construct the lower-triangular Toeplitz matrix $\hat \strategy$ with first column $\hat c_1, \hat c_2, \ldots$.
\end{itemize}
Moreover, this holds for any leading principal sub-matrix: that is, $\strategy[1:n, 1:n]^{-1} = \hat \strategy[1:n, 1:n]$ for any integer $n > 0$. Thus, it suffices to consider infinite Toeplitz matrices; we denote infinite BLT matrices as $\BLT(\scalev, \decayv)$ by dropping the subscript $n$.

\paragraph{BLT to Generating Function}
We start by computing the generating function of the BLT  and inverse BLT:
\begin{lemma} \label{lem:blt-ogfs}
    The generating function $\genf(x)$ of $\BLT(\scalev, \decayv)$ is given by
    \begin{align} \label{eq:genfn-c1}
        \genf(x) &= 1 + x \, \frac{\polyDiff(x)}{\polyDenom(x)} = \frac{\polyInvDenom(x)}{\polyDenom(x)} \,,
    \end{align}
    where we define the polynomials
    \begin{align}  \label{eq:genfn-c2}
        \polyDenom(x) = \prod\iind (1 - \decay_i x), \quad
        \polyDiff(x) = \sum\iind \frac{\scale_i \polyDenom(x)}{1 - \decay_i x},
        \quad \text{and} \quad
        \polyInvDenom(x) = \polyDenom(x) + x \, \polyDiff(x) \,.
    \end{align}
    Moreover, the generating function of $\hatgenf$ of the inverse matrix $\BLT(\scalev, \decayv)^{-1}$ is given by 
    \begin{align} \label{eq:genfn-cinv}
    \begin{aligned}
        \hatgenf(x) = \frac{\polyDenom(x)}{\polyInvDenom(x)} = 1 - x \,\frac{\polyDiff(x)}{\polyInvDenom(x)} \,.
    \end{aligned}
    \end{align}
\end{lemma}
\begin{proof}
By summing the geometric series, we get
\[
    \genf(x) = 1 + \sum_{t=1}^\infty \left(\sum\iind \scale_i \decay_i^{t-1} \right) x^t = 1 + \sum\iind \frac{\scale_i x}{1 - \decay_i x} \,.
\]
Simplifying this gives \cref{eq:genfn-c1}. For the inverse, we have from \Cref{lem:generating-fn} that
\[
     \hatgenf(x) = \frac{1}{\genf(x)} = \frac{\polyDenom(x)}{\polyDenom(x) + x \, \polyDiff(x)}
     = 1 - \frac{x \, \polyDiff(x)}{\polyDenom(x) + x\, \polyDiff(x)} = 1 + x\, \frac{-\polyDiff(x)}{\polyInvDenom(x)} \,.
\]
\end{proof}

Note that the polynomial $\polyDenom$ is of degree $d$, while the polynomial $\polyDiff$ is of degree $d-1$. The polynomial $\polyInvDenom(x) = \polyDenom(x) + x\, \polyDiff(x)$ is thus a sum of two degree-$d$ polynomials, and its degree $D = \mathrm{deg}(\polyInvDenom)$ can be at most $d$.
The generating functions $\genf$ and $\hatgenf$ are rational functions of degree at most $d$ in both the numerator and denominator.

\paragraph{Generating Function to Inverse BLT}
To reconstruct the inverse BLT, we need to find the Maclaurin series of its generating function $\hatgenf$. When $\hatgenf$ is a rational function, as is the case in \Cref{lem:blt-ogfs}, this is most easily obtained by a partial fraction decomposition of $\hatgenf(x)$. This approach is also commonly used in solving recursions in combinatorics and discrete mathematics, and filter design in digital signal processing.

The first step to construct a partial fraction decomposition is to reason about the roots of the denominator $\polyInvDenom(x) = \polyDenom(x) + x \, \polyDiff(x)$.
\begin{proposition} \label{prop:roots}
     Consider the setting of \Cref{thm:main} with parameters $\scalev \in \R^d_{++}$ and $\decayv \in (0, 1)^d$, and let $\polyDiff(x)$ and $\polyDenom(x)$ be as defined in \cref{eq:genfn-c2} (see \Cref{lem:blt-ogfs}). Then, the polynomial $\polyInvDenom(x) = \polyDenom(x) + x \, \polyDiff(x)$ has degree $D=d-1$ if $\sum\iind \scale_i / \decay_i = 1$ and $D = d$ otherwise. Moreover, all its roots $\nu_1, \ldots, \nu_D$ are unique and real.
\end{proposition}

Since the polynomial $\polyInvDenom(x) = \polyDenom(x) + x \, \polyDiff(x)$ has all unique and real roots, the partial fraction decomposition of $\polyDiff(x) / \polyInvDenom(x)$ takes a very simple form:

\begin{proposition} \label{prop:partial-frac}
    Consider the setting of \Cref{prop:roots}. 
    The polynomials $\polyDiff$ and $\polyInvDenom$ are co-prime and we have a unique partial fraction decomposition
    \begin{align} \label{eq:partial-frac}
        \frac{\polyDiff(x)}{\polyInvDenom(x)} = -\sum\iind \frac{\hat \scale_i}{1 - \hat \decay_i x} 
    \end{align}
    for some $\hat\scalev, \hat \decayv \in \R^d$ with $\hat\scalev$ non-zero component-wise.
    Further, we have the following based on $\mathrm{deg}(\polyInvDenom) = D$:
    \begin{enumerate}[label=(\alph*)]
        \item If $D = d$, then $\hat \decay_i \neq 0$ for all $i \in [d]$.
        \item \label{item:pfd-low-rank}
        If $D = d-1$, then $\hat \decay_i \neq 0$ for $i \in [d-1]$ and $\hat \decay_d = 0$.
        \item Finally, we have $\hat\decay_i = 1/\nu_i$ for $i \in [D]$ where $\nu_1, \ldots, \nu_D$ are the roots of the degree-$D$ polynomial $\polyInvDenom$.
    \end{enumerate}
\end{proposition}

\noindent Together with \Cref{lem:generating-fn,lem:blt-ogfs}, this immediately implies a partial fraction decomposition of the generating function $\hatgenf$ of $\BLT(\scalev, \decayv)\inv$:
\begin{corollary} \label{cor:blt-inverse-form}
    In the setting of \Cref{prop:partial-frac}, we have that the generating function of $\BLT(\scalev, \decayv)\inv$ is given by
    \begin{align} \label{eq:hatgenf}
    \hatgenf(x) = 1 + x\,\frac{-\polyDiff(x)}{\polyInvDenom(x)} = 1 + \sum\iind \frac{\hat\scale_i x}{1-\hat \decay_i x}\,,
    \end{align}
    where $\hat \scalev, \hat \decayv$ are as given in \Cref{prop:partial-frac}. In particular, we have that $\BLT(\scalev, \decayv)\inv = \BLT(\hat\scalev, \hat\decayv)$.
\end{corollary}

\Cref{lem:blt-ogfs-new} of \Cref{sec:notation} conveniently summarizes all the generating function results that we have established so far.

\paragraph{Inverse BLT Parameter Properties}
In order to complete the proof, we must argue about the signs and magnitudes of $\hat\scale_i$'s and $\hat \decay_i$'s depending on the value of $\sum_i \scale_i / \decay_i$. The main technical result we show is:

\begin{proposition} \label{prop:technical-main}
    Consider the setting of \Cref{thm:main} with parameters $\scalev \in \R^d_{++}$ and $\decayv \in (0, 1)^d$ distinct, and let $\polyDiff(x)$ and $\polyDenom(x)$ be as defined in \cref{eq:genfn-c2} (see \Cref{lem:blt-ogfs}). Then, we have the following:
    \begin{enumerate}[label=(\alph*)]
        \item  \label{item:prop-tech-main:good-roots}
        If $\sum\iind \scale_i /\decay_i < 1$, then all roots $\nu_1, \ldots, \nu_d$ of $\polyInvDenom$ lie in $(1, \infty)$. Thus, the parameter $\hat\decay_i = 1 / \nu_i$ in the denominator of the partial fraction decomposition \eqref{eq:partial-frac} lies in $(0, 1)$ for each $i=1, \ldots, d$.
        \item  \label{item:prop:neg-root}
         If $\sum\iind \scale_i /\decay_i > 1$, then one root of $\polyInvDenom$ lies in $(-\infty, -1)$ while all other roots lie in $(1, \infty)$.
         Thus, we have $\hat\decay_i \in (0, 1)$ for $i=1, \ldots, d-1$ and $\hat \decay_d \in (-1, 0)$.
        \item  \label{item:prop-tech-main:good-consts}
         Irrespective of the value of $\sum\iind \scale_i/\decay_i$, we have that the numerator of the the partial fraction decomposition \eqref{eq:partial-frac} satisfies $\hat\scale_i < 0$ for $i=1, \ldots, d$.
    \end{enumerate}
\end{proposition}

\noindent Given \Cref{prop:partial-frac,prop:technical-main}, the proof of \Cref{thm:main} is immediate.

\begin{proof}[Proof of \Cref{thm:main}]
    The matrix $\BLT_n(\scalev, \decayv)$ is a lower-triangular matrix with all ones on the diagonal; thus it is invertible for all $n > 0$. By \Cref{cor:blt-inverse-form}, we have that 
    $\BLT(\scalev, \decayv)\inv = \BLT(\hat\scalev, \hat\decayv)$, where $\hat\scalev, \hat\decayv$ are as in \Cref{prop:partial-frac}. Then, the signs and magnitudes of $\hat\scalev, \hat \decayv$ imply the various parts of \Cref{thm:main}.
    In particular, part~\ref{part:thm:lt-1} of \Cref{thm:main} follows directly from from part~\ref{item:prop-tech-main:good-roots} of \Cref{prop:technical-main}.
    Similarly, part~\ref{part:thm:gt-1} of \Cref{thm:main} follows from \Cref{prop:technical-main}\ref{item:prop:neg-root}.
    Next, \Cref{thm:main}\ref{part:thm:eq-1} follows from 
    \Cref{prop:partial-frac}\ref{item:pfd-low-rank}, while the negative scale parameters follows from \Cref{prop:technical-main}\ref{item:prop-tech-main:good-consts}.
    Finally, the uniqueness of the inverse BLT parameters follows from two observations:
    \begin{itemize}
    \item the decay parameters $\hat \decayv$ are obtained as the roots of the polynomial $\polyInvDenom(x) = \polyDenom(x) + x\, \polyDiff(x)$ and are unique.
    \item the scale parameters $\hat \scalev$ are the coefficients of a partial fraction and are unique as per \Cref{prop:partial-frac}.
    \end{itemize}
\end{proof}

\begin{figure}[p]
    \centering
    \adjustbox{max height=0.75\textheight}{
    \includegraphics[width=0.99\linewidth]{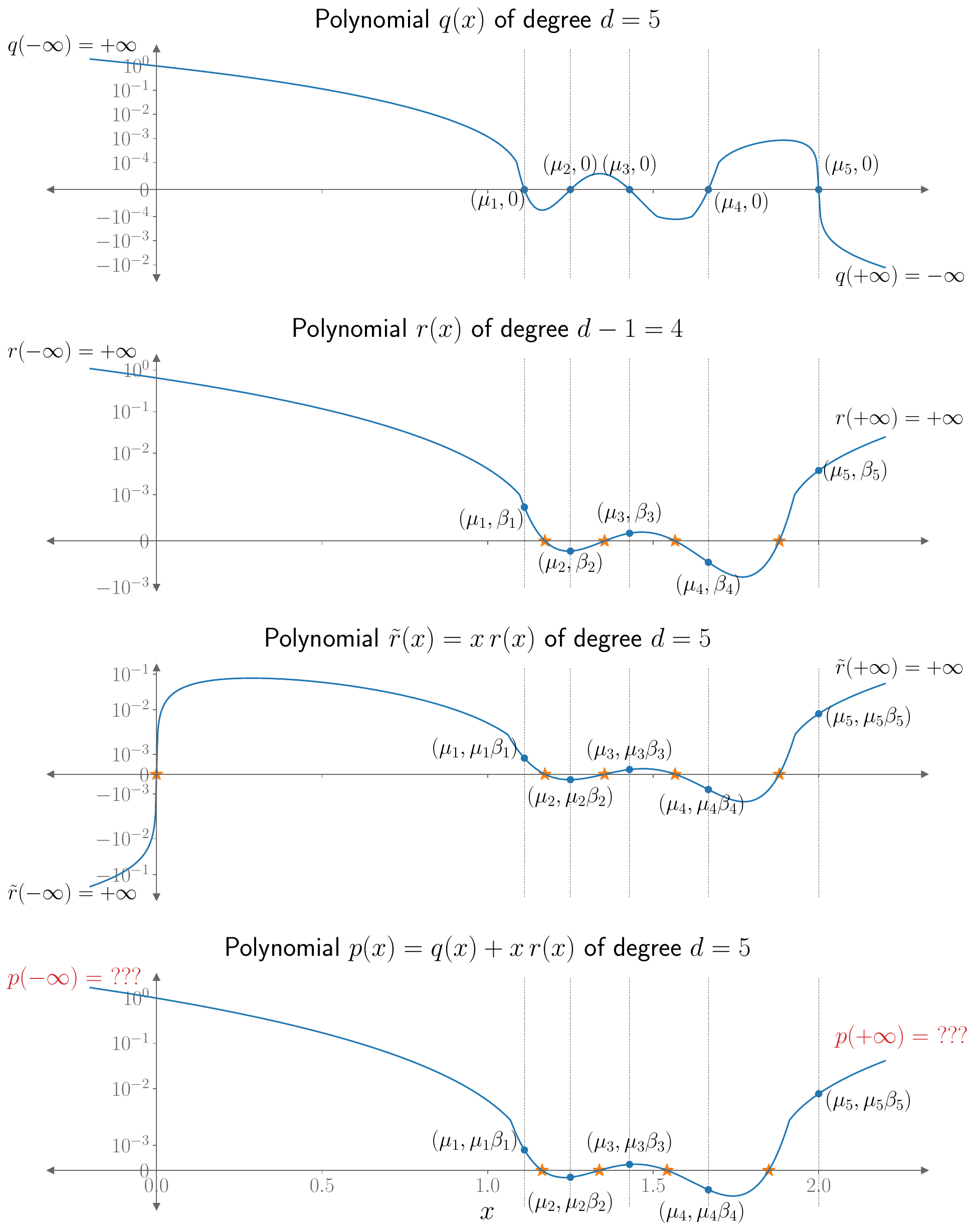}}
    \caption{\small {Illustrations of the polynomials $\polyDiff, \polyDenom, \polyInvDenom$ for $d=5$ in symmetrical log scale}.
    \textbf{First row}: Let $\mu_i 
    := 1/\decay_i > 1$ for $i=1, \ldots, d$ denote the roots of $\polyDenom(x)$ in ascending order. \textbf{Second row}: We show that $\beta_i := \polyDiff(\mu_i)$ is positive for $i$ odd and negative for $i$ even. Due to sign changes, each of the $d-1$ roots of $\polyDiff$ lies between $(\mu_i, \mu_{i+1})$ for some $i$ (denoted by the orange star). \textbf{Third row}: By the same argument, each $d-1$ non-zero roots of $x\, \polyDiff(x)$ lie in $(\mu_i, \mu_{i+1})$ for some $i$. 
    \textbf{Last row}: The same argument accounts for $d-1$ roots of the degree-$d$ polynomial $\polyInvDenom(x) = \polyDenom(x) + x \,\polyDiff(x)$. Since $d-1$ roots of the degree-$d$ real polynomial $\polyInvDenom$ are real, the final root must be real as well, establishing \Cref{prop:roots}. This is continued in \Cref{fig:deg-5b}; a similar argument also works $d$ even---see \Cref{fig:deg-4}.
    }
    \label{fig:deg-5}
\end{figure}

\begin{figure}[t]
    \centering
    \includegraphics[width=0.7\linewidth]{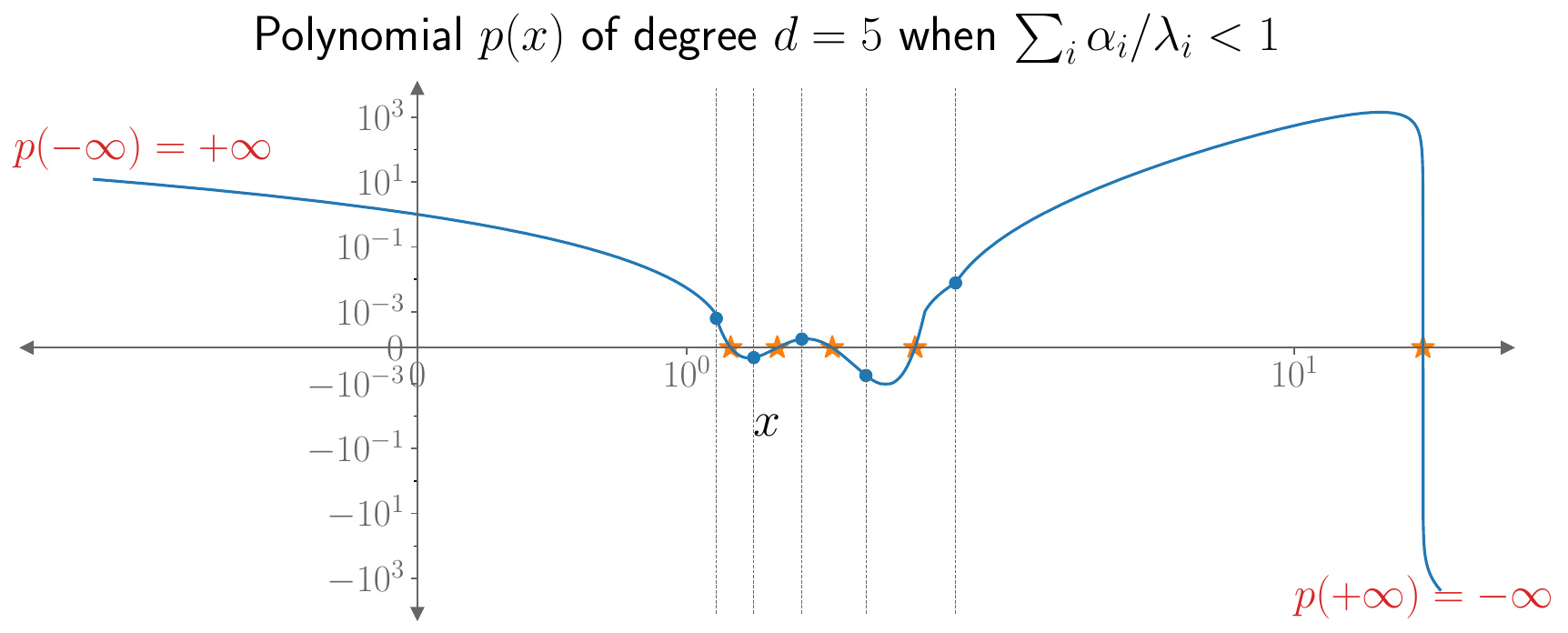}
    \includegraphics[width=0.7\linewidth]{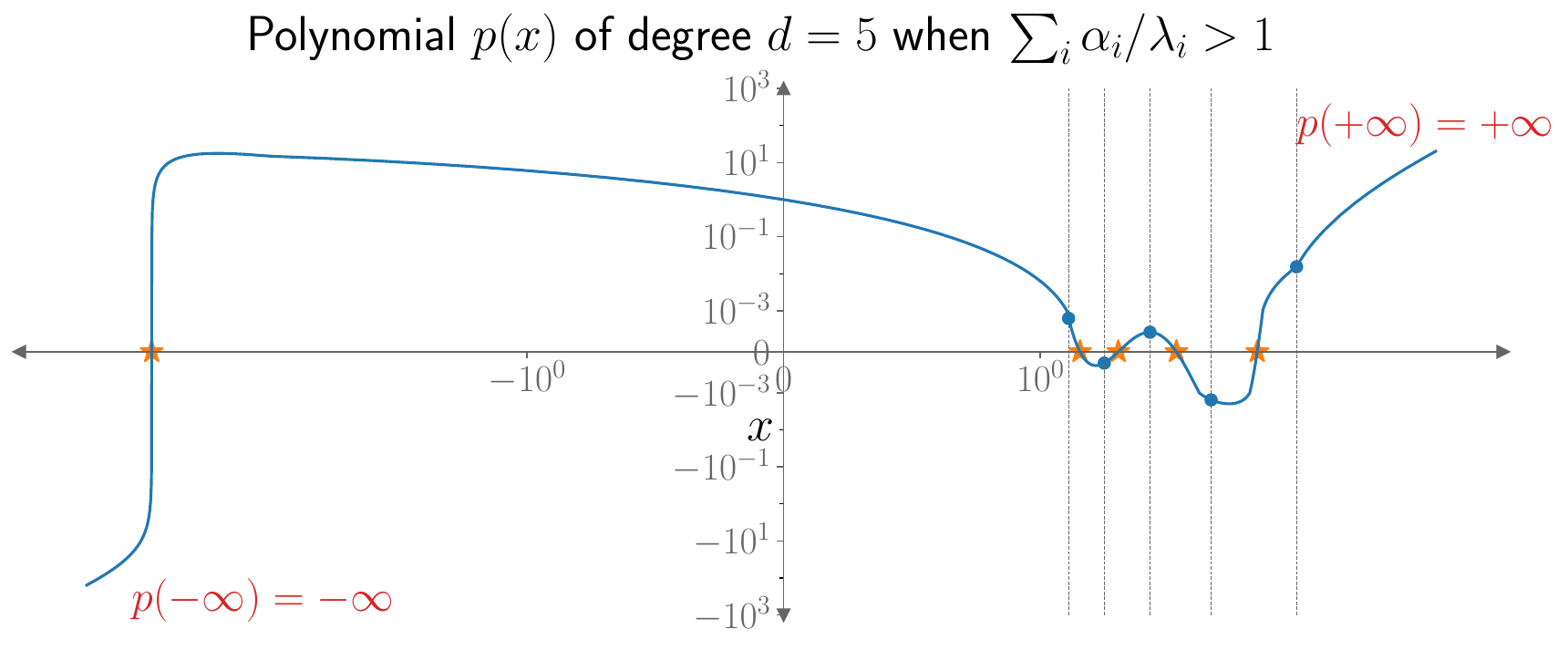}
    \caption{\small Continued from \Cref{fig:deg-5}, which shows $d-1$ roots of $\polyInvDenom(x) = \polyDenom(x) + x\,\polyDiff(x)$ for an example with $d=5$. This figure illustrates how the final $d$\textsuperscript{th} root of $\polyInvDenom(x)$ depends on the BLT parameters $\scalev, \decayv$. As previously, the dotted lines denote the roots $\mu_1, \ldots, \mu_d$ of $\polyDenom(x)$ (where $\mu_i = 1/\decay_i$) and the orange stars denote the roots of $\polyInvDenom(x)$.
    \textbf{Top}: When $\sum\iind \scale_i/  \decay_i < 1$, the last root of $\polyInvDenom$ is positive as well, as in \Cref{prop:technical-main}\ref{item:prop-tech-main:good-roots}.
    \textbf{Bottom}: When $\sum\iind \scale_i/  \decay_i > 1$, then $\polyInvDenom$ has a  negative root, as in \Cref{prop:technical-main}\ref{item:prop:neg-root}.
    }
    \label{fig:deg-5b}
\end{figure}
        
We summarize the key ideas behind the proofs of \Cref{prop:roots,prop:partial-frac,prop:technical-main}, with full proofs appearing in \Cref{sec:proof} (see \Cref{fig:deg-5,fig:deg-5b} for a key idea):
\begin{itemize}
    \item Let $\mu_1 < \cdots < \mu_d$ denote the roots of the polynomial $\polyDenom(x) = \prod\iind (1- \decay_i x)$. Since $\decay_i \in (0, 1)$, we have $\mu_i = \decay_i^{-1} \in  (1, \infty)$. Notably, all these roots are real.
    \item The key step in the proof is to argue about the sign of $\beta_i := \polyDiff(\mu_i)$. We show in the upcoming \Cref{prop:p_of_mu_i} that $\beta_i > 0$ for $i$ odd and $\beta_i < 0$ for $i$ even.
    \item Thus, we have that $\polyInvDenom(\mu_i) = \polyDenom(\mu_i) + \mu_i\, \polyDiff(\mu_i) = \mu_i \beta_i$ is positive for $i$ odd and negative for $i$ even. Thus, $\polyInvDenom$ admits a root in the interval $(\mu_i, \mu_{i+1})$ for each $i = 1, \ldots, d$. This accounts for $d-1$ roots of $\polyInvDenom$, which are real and positive. 
    Thus, $\polyInvDenom$ is of degree $d-1 \le D \le d$. 
    If $\polyInvDenom$ is of degree $D=d$, then the $d$\textsuperscript{th} root of $\polyInvDenom$ is real as well, since complex roots of a real polynomial can only occur in conjugate pairs---this gives \Cref{prop:roots}. 
    \item Next, reasoning about the partial fraction decomposition of the ordinary generating function $\hatgenf$ of the inverse BLT (defined in \Cref{lem:blt-ogfs}) gives us \Cref{prop:partial-frac}. In view of \Cref{cor:blt-inverse-form}, all that is now left for \Cref{prop:technical-main} is to reason about the final root of $\polyInvDenom$ (it can either be smaller than 0 or larger than 1) as well as the $\hat \scale_i$ coefficients.
    \item Next, 
    we argue about the last root of $\polyInvDenom$, in the case that its degree is $D=d$. The leading coefficients of $\polyDenom(x)$ and $x\,\polyDiff(x)$ have opposite signs; see \Cref{fig:deg-5}. We calculate the coefficient $\polyInvDenom_d$ of $x^d$ in the polynomial $\polyInvDenom$ as
    \[
        \polyInvDenom_d = (-1)^d \left( \prod\iind \decay_i \right) \left( 1- \sum\iind \frac{\scale_i}{\decay_i} \right) \,.
    \]
    Thus, we get the following cases (see \Cref{fig:deg-5b}):
    \begin{enumerate}[label=(Case \Roman*), leftmargin=\widthof{ (Case IIIV)}]
        \item If $\sum\iind \scale_i / \decay_i < 1$, then $\polyInvDenom(-\infty) = +\infty$. We show that $\polyInvDenom(\mu_d)$ and $\polyInvDenom(\infty)$ have opposite signs.
         Thus, the $d$\textsuperscript{th} root of the polynomial $\polyInvDenom$ is also larger than 1.
        \item If $\sum\iind \scale_i / \decay_i > 1$, then $\polyInvDenom(-\infty) = -\infty$. We show that $\polyInvDenom(-1) > 0$, leading to the conclusion that the $d$\textsuperscript{th} root of the polynomial $\polyInvDenom$ lies in $(-\infty, -1)$.
        \item If $\sum\iind \scale_i / \decay_i =1$, the leading order terms of $\polyDenom(x)$ and $x\, \polyDiff(x)$ cancel out and $\polyInvDenom(x)$ is a polynomial of degree $d-1$. Thus, it only has $d-1$ roots, all of which have previously been accounted for.
    \end{enumerate}
    This yields parts \ref{item:prop-tech-main:good-roots} and \ref{item:prop:neg-root} of \Cref{prop:technical-main}.
    \item Finally, we need to argue that the scale parameters $\hat\scale_i$ are negative. We do so by arguing about their signs from their closed-form expressions.
\end{itemize}

See \Cref{sec:proof} for full proofs of each of these steps.

\paragraph{Proof of \Cref{thm:calc-output-scale}}
We conclude this section with a proof of \Cref{thm:calc-output-scale}.
\begin{proof}[Proof of \Cref{thm:calc-output-scale}]
    The expressions of \cref{eq:calc-output-scale} can be obtained by simplifying Eq. (5.2) of \citet[Lemma 5.2]{dvijotham2024efficient}. Alternatively, \Cref{prop:consts-expression} of \Cref{sec:proof} gives a short and elementary proof of the expression for $\hat \scale_i$. Then, the expression for $\scale_i$ can be obtained by symmetry.
    
    Next, we turn to the equivalence of various representations.
    \begin{itemize}
        \item \ref{item:pos-param} $\implies$ \ref{item:decay-param}: Assuming w.l.o.g. the ordering $\decay_1 > \cdots \decay_d$, \Cref{thm:main}\ref{part:thm:lt-1} gives us unique $\hat \scalev < 0$ and $\hat \decayv \in (0, 1)^d$. The proof that $\decayv, \hat \decayv$ satisfy the claimed ordering is given in \Cref{cor:order-of-decays} in \Cref{sec:proof}.
        \item \ref{item:decay-param} $\implies$ \ref{item:pos-param}: We can evaluate the signs of the scale parameters $\scale_i$ in \Cref{eq:calc-output-scale}. For $\scale_1$, we have that all the terms in \cref{eq:calc-output-scale} are positive. For $i=2$, we have only one negative term $(\decay_2 - \hat \decay_1)$ in the numerator and only one negative term $(\decay_2 - \decay_1)$ in the denominator, and is thus positive. Similarly, each $\scale_i$ has $i-1$ negative terms each in the numerator and denominator, so that $\scale_i > 0$. To show the bound on $\sum_i \scale_i / \decay_i$, rearranging \cref{eq:pdf-identity} of \Cref{prop:consts-expression} (with $(\scalev, \decayv)$ and $(\hat \scalev, \hat \decayv)$ swapped) gives:
        \[
            \sum\iind \frac{\scale_i}{\decay_i} = 1 - \frac{\prod\iind \hat \decay_i}{\prod\iind \decay_i} < 1 \,,
        \]
        since $0 < \hat \decay_i / \decay_i < 1$ by assumption.
    \end{itemize}
\end{proof}

\section{Full Proof Details}
\label{sec:proof}

We give the full proofs of all remaining statements from \Cref{sec:proof-main}, i.e., \Cref{prop:roots,prop:partial-frac,prop:technical-main}.
In particular, \Cref{prop:roots} from \Cref{sec:proof-main} is a special case of \Cref{prop:r_of_mu_i,prop:final-root}\ref{item:final-root-low-rank}, while \Cref{prop:partial-frac} is established in Step 4 below and \Cref{prop:technical-main} follows directly from \Cref{prop:final-root} and \Cref{prop:consts} below.
The assumptions of \Cref{thm:main} are assumed to hold throughout. We also fix the BLT parameters $\scalev\in\R^d_{++}$ and $\decayv\in(0, 1)^d$ throughout, except when explicitly mentioned otherwise (e.g. \Cref{prop:consts-expression} is a notable exception).

\paragraph{Step 1: Notation and Properties of $\polyDenom(x)$}
As introduced previously, define $\mu_i = \decay_i^{-1}$ for $i=1,\ldots, d$. These are the roots of the polynomial $\polyDenom(x)$ from \cref{eq:genfn-c2}. We assume that $\mu_1 < \mu_2 < \cdots < \mu_d$; this is without loss of generality as we assumed in \Cref{thm:main} that $\decay_i$'s are distinct.
Finally, define the constant
\begin{align} \label{eq:M}
    M = \prod\iind \mu_i \,.
\end{align}
Using this, we can rewrite $\polyDenom(x)$ from \cref{eq:genfn-c2} as 
\begin{align} \label{eq:q-rewrite}
    \polyDenom(x) = \frac{(-1)^d}{M} \prod\iind (x - \mu_i) \,.
\end{align}

\paragraph{Step 2: Behavior of $\polyDiff(x)$ at the roots of $\polyDenom(x)$}
Let $\beta_i := \polyDiff(\mu_i)$ be the values of the polynomial $\polyDiff$ at the roots $\mu_1, \ldots, \mu_d$ of $\polyDenom$. (See \cref{eq:genfn-c2} for the definition of $\polyDiff$.)
We now argue that $\beta_i$ is positive if $i$ is odd and negative otherwise.
\begin{property} \label{prop:p_of_mu_i}
    We have
    \[
        \beta_i := \polyDiff(\mu_i) = \frac{\scale_i \mu_i}{M} \prod_{j \neq i}(\mu_j - \mu_i) \,.
    \]
    In particular, $\beta_i > 0$ if $i$ is odd and $\beta_i < 0$ if $i$ is even.
\end{property}
\begin{proof}
    Starting from the definition of the degree-$(d-1)$ polynomial $\polyDiff$ from \cref{eq:genfn-c2}, we have,
    \begin{align} \label{eq:p-rewrite}
        \polyDiff(x) &= \sum\iind \scale_i \prod_{j\neq i}(1 - \decay_j x) = \sum\iind \frac{\scale_i\mu_i}{M} \prod_{j\neq i} (\mu_j - x) \,,
    \end{align}
    where we substituted $\mu_i = \decay_i^{-1}$ and $M = \prod\iind \mu_i$.
    When computing $\polyDiff(\mu_i)$, we note that all but the $i$\textsuperscript{th} term will be zero, yielding the claimed expression for $\beta_i$. 
    Next, we turn to the signs: 
    $\beta_i$ has $(i-1)$ negative terms in the product, so its sign is the same as $(-1)^{i-1}$.
\end{proof}

\begin{remark}
    The polynomial $\polyDiff(x)$ can also be interpreted through the lens of Lagrange interpolation. We can rewrite \cref{eq:p-rewrite} as
    \[
        \polyDiff(x) = \sum\iind \beta_i \prod_{j \neq i} \frac{x - \mu_j}{\mu_i - \mu_j} \,,
    \]
    which is the Lagrange interpolant through $(\mu_1, \beta_1), \ldots, (\mu_d, \beta_d)$. In fact, $\polyDiff(x)$ is the unique degree $(d-1)$ polynomial interpolating these $d$ points. This fact was used in the proof of \cite[Lemma 5.2]{dvijotham2024efficient}, but we prove \Cref{prop:p_of_mu_i} (and the upcoming \Cref{prop:consts-expression}) by more direct and elementary means.
\end{remark}

\paragraph{Step 3: Roots of $\polyInvDenom(x)$}
The polynomial $\polyInvDenom(x) = \polyDenom(x) + x\, \polyDiff(x)$ is of degree $D\le d$. Let $\nu_1, \ldots, \nu_D$ be the roots of the polynomial (if they exist). We establish \Cref{prop:roots} with some additional properties which will be useful later:

\begin{property} \label{prop:r_of_mu_i}
    The polynomial $\polyInvDenom(x) = \polyDenom(x) + x\, \polyDiff(x)$ is of degree $D \in \{d-1, d\}$ are all its roots are unique and real. Further, there exists a root $\nu_i$ of $\polyInvDenom$ in the interval $(\mu_i, \mu_{i+1})$ for $i=1, \ldots, d-1$.
\end{property}
\begin{proof}
    Recall that the $\mu_i > 1$ are the roots of $\polyDenom(x)$. From \Cref{prop:p_of_mu_i}, we deduce that $\polyInvDenom(\mu_i) = \polyDenom(\mu_i) + \mu_i \, \polyDiff(\mu_i) = \mu_i \beta_i$ is positive for $i$ odd and negative for $i$ even. Next, we invoke the intermediate value theorem: since $\polyInvDenom(\mu_i)$ and $\polyInvDenom(\mu_{i+1})$ have opposite signs and $\polyInvDenom$ is a polynomial (and thus continuous), there exists $\nu_i \in (\mu_i, \mu_{i+1})$ such that $\polyInvDenom(\nu_i) = 0$ for each $i=1, \ldots, d-1$. 
    
    Finally, since $d-1$ roots exist, the polynomial $\polyInvDenom$ can only be of degree $d$ or $d-1$. If $\polyInvDenom$ is of degree $D=d-1$, there is nothing left to prove, so suppose that $\polyInvDenom$ is of degree $D=d$. Since $d-1$ roots of a degree-$d$ real polynomial are real, then the final $d$\textsuperscript{th} root should be real as well---this is because complex roots can only occur in conjugate pairs.
\end{proof}

It remains to argue about the potential final root of the polynomial $\polyInvDenom$:

\begin{property} \label{prop:final-root}
    In the setting of \Cref{prop:r_of_mu_i}, we have the following:
    \begin{enumerate}[label=(\alph*)]
        \item \label{item:final-root-low-rank}
        If $\sum\iind {\scale_i}/{\decay_i} = 1$, then $\polyInvDenom$ is a polynomial of degree $D= d-1 $. Otherwise, it is of degree $D=d$.
        \item \label{item:final-root-lt-1}
        If $\sum\iind {\scale_i}/{\decay_i} < 1$, then, the $d$\textsuperscript{th} root $\nu_d$ of $\polyInvDenom$ satisfies $\nu_d > \mu_d > 1$.
        \item \label{item:final-root-gt-1}
        If $\sum\iind {\scale_i}/{\decay_i} > 1$ but $\sum\iind \scale_i < 1$, then, the $d$\textsuperscript{th} root $\nu_d$  of $\polyInvDenom$ satisfies $\nu_d < -1$.
    \end{enumerate}
\end{property}
\begin{proof}
    We reason about the leading coefficient $\polyInvDenom_d$ of $x^d$ in the polynomial $\polyInvDenom(x) = \polyDenom(x) + x\, \polyDiff(x)$ such that $\polyInvDenom(x) - \polyInvDenom_d x^{d}$ is a polynomial of degree $(d-1)$.
    From \cref{eq:q-rewrite,eq:p-rewrite}, we deduce that
    \[
        \polyInvDenom_d = \frac{(-1)^d}{M} \left( 1- \sum\iind \scale_i \mu_i \right) = \frac{(-1)^d}{M} \left( 1- \sum\iind \frac{\scale_i}{\decay_i} \right) \,.
    \]
    Part \ref{item:final-root-low-rank} follows because $\polyInvDenom_d = 0$ if and only if $\sum\iind {\scale_i}/{\decay_i}=1$.
    Next, suppose that $\sum\iind {\scale_i}/{\decay_i} < 1$. We have two cases:
    \begin{itemize}
        \item $d$ is odd (see \Cref{fig:deg-5b} for reference): We have that $\polyInvDenom(\mu_d) > 0$ and $\lim_{x \to \infty} \polyInvDenom(x) = -\infty$ (because $\polyInvDenom_d < 0$). 
        \item $d$ is even (see \Cref{fig:deg-4b} for reference): We have that $\polyInvDenom(\mu_d) < 0$ and $\lim_{x\to \infty} \polyInvDenom(x) > 0$ (because $\polyInvDenom_d > 0$). 
    \end{itemize}
    In both cases, by the intermediate value theorem, the final root $\nu_d$ of $\polyInvDenom$ must lie in $(\mu_d, \infty)$, yielding part~\ref{item:final-root-lt-1}. 
    
    The proof of part~\ref{item:final-root-gt-1} also proceeds similarly. We have $\lim_{x \to -\infty} \polyInvDenom(x) < 0$ both for even and odd degree $d$. On the other hand, we have
    \begin{align*}
        \polyInvDenom(-1) &= \polyDenom(-1) - \polyDiff(-1)
              = \prod\iind (1 + \decay_i) \left( 1 - \sum\iind \frac{\scale_i}{1 + \decay_i}\right) > 0 \,,
    \end{align*}
    where the last inequality followed from $\scale_i, \decay_i > 0$ and
    \[
        \sum\iind \frac{\scale_i}{1 + \decay_i} < \sum\iind \scale_i < 1 
    \]
    by assumption.
    Another invocation of the intermediate value theorem implies that $\nu_d \in (-\infty, -1)$, completing the proof.
\end{proof}

We summarize the ordering of the decay parameters $\decayv, \hat \decayv$:
\begin{corollary} \label{cor:order-of-decays}
    In the setting of \Cref{prop:technical-main}, we have 
    \[
        \decay_1 > \hat \decay_1 > \decay_2 > \hat \decay_2 > \cdots > \hat \decay_{d-1} > \decay_d > \hat \decay_d \,.
    \]
\end{corollary}
\begin{proof}
     Recall that we assumed the roots $\mu_i = \decay_i^{-1}$ of $\polyDenom(x)$ are ordered as $\mu_1 <  \cdots < \mu_d$. \Cref{prop:roots,prop:final-root} tell us that the roots $\nu_1, \ldots, \nu_d$ of $\polyInvDenom(x)$ satisfy $\mu_i < \nu_i < \mu_{i+1}$ for $i = 1, \ldots, d-1$. Thus, we have that $\decay_i = \mu_i^{-1}$ for $i \in [d]$ and $\hat \decay_i = \nu_i^{-1}$ for $i \in [d]$ are ordered as claimed. 
     If $\sum_i \scale_i / \decay_i < 1$ (as in \Cref{prop:technical-main}\ref{item:prop-tech-main:good-roots}), we have that $\nu_d > \mu_d$, leading to $\decay_d > \hat \decay_d$. 
     If not, we have $\hat \decay_d \le 0$ by \Cref{prop:final-root} again and $\decay_d > 0$ by assumption, leading to the claimed order.
\end{proof}

\paragraph{Step 4: Partial Fraction Decomposition}
We now give the proof of \Cref{prop:partial-frac} from \Cref{sec:proof-main}.

\begin{proof}[Proof of \Cref{prop:partial-frac}]
    We first prove that the polynomials $\polyDiff$ and $\polyInvDenom$ are co-prime, meaning that they do not share any roots.\footnote{
    Note that if $\polyDiff$ and $\polyInvDenom$ are co-prime, then the rational function $\polyDiff/\polyInvDenom$ is irreducible.
    } 
    Note that $\polyDiff$ and $\polyDenom$ are co-prime because $\polyDiff(\mu_i)$ is non-zero as per \Cref{prop:p_of_mu_i}, where $\mu_i$'s are the roots of $\polyDenom$. Then, we get that the greatest common divisor (denoted ``gcd'') of $\polyDiff$ and $\polyInvDenom$ is 
    \[
        \mathrm{gcd}(\polyDiff(x), \polyInvDenom(x)) = \mathrm{gcd}(\polyDiff(x), \polyDenom(x) + x\, \polyDiff(x))
        = \mathrm{gcd}(\polyDiff(x), \polyDenom(x)) = 1 \,,
    \]
    where the second equality used the property that $\mathrm{gcd}(\polyDiff, \polyDenom) = \mathrm{gcd}(\polyDiff, \polyDenom + \varphi\, \polyDiff)$ for any polynomial $\varphi$ (we take $\varphi(x) = x$). Thus, the $\polyDiff(x) / \polyInvDenom(x)$ is the ratio of of degree $d-1$ polynomial to a degree $D \in \{d-1, d\}$ polynomial.
    Using the fact that the roots $\nu_1, \ldots, \nu_D$ of the polynomial $\polyInvDenom$ are real and unique by \Cref{prop:roots}, we have the general form of the partial fraction decomposition
    \begin{align}
        \frac{\polyDiff(x)}{\polyInvDenom(x)} = \kappa_0 + \sum\iinD \frac{\kappa_i}{x - \nu_i} \,,
    \end{align}
    for some reals $\kappa_0, \kappa_1, \ldots, \kappa_D$. 
    Note that $\kappa_1, \ldots, \kappa_D$ are non-zero, while $\kappa_0$ is allowed to take zero values.
    Indeed, if any of $\kappa_1, \ldots, \kappa_D$ were zero, we would not obtain the correct degree in the denominator on the left side.
    We consider two separate cases depending on the degree $D$ of the polynomial $\polyInvDenom$.
    \begin{itemize}
        \item Case $D=d$: If $\kappa_0 \neq 0$, then numerator on the right side would have degree $d$. This is a contradiction because the degree of the numerator $\polyDiff$ on the left side is $d-1$. Thus, we must have $\kappa_0 = 0$. In this case, we set $\hat \decay_i = 1 / \nu_i$ and $\hat \scale_i = \kappa_i / \nu_i$ for $i \in [d]$ to obtain \eqref{eq:partial-frac}---note that all of these are non-zero.
        \item Case $D = d-1$: If $\kappa_0 = 0$, then the numerator on the right side would have degree $D-1=d-2$, which is again a contradiction (since $\mathrm{deg}(\polyDiff)=d-1$ on the left side). Thus, we must have $\kappa_0 \neq 0$. In this case, we set $\hat \decay_i = 1 / \nu_i$ and $\hat \scale_i = \kappa_i / \nu_i$ for $i \in [d-1]$---note that all of these are non-zero. We set $\hat \decay_d = 0$ and set the constant term $\hat \scale_d = -\kappa_0$ to obtain \cref{eq:partial-frac}, as desired.
    \end{itemize}
    Further, since the factors of the denominator are linear and unique (since $\polyInvDenom(x)$ has no repeated roots), there exist unique coefficients $\kappa_0, \ldots, \kappa_D$ that satisfy the partial fraction decomposition.
\end{proof}

\paragraph{Step 5: Scale Parameters of the Inverse BLT}
It remains to show \Cref{prop:technical-main}\ref{item:prop-tech-main:good-consts} regarding the sign of the scale parameters.
We start with a self-contained proof of the expression for the scale parameters from \Cref{thm:calc-output-scale}, and then argue about its sign.

\begin{lemma} \label{prop:consts-expression}
    Let $\decayv, \hat \decayv \in (\R \setminus \{0\})^d$ be distinct non-zero vectors (i.e., $\decay_i \neq \decay_j$ and $\hat \decay_i \neq \hat \decay_j$ for all $i \neq j$).\footnote{
        Note that (a) we do not restrict any of the decay parameters to lie in $(0, 1)$, and (b) the assumptions allow us to swap the roles of $(\decayv, \scalev)$ and $(\hat \decayv, \hat \scalev)$.
    } 
    Suppose also that $\decay_i \neq \hat\decay_j$ for all $i, j \in [d]$.
    Then, the constants $\hat \scale_1, \ldots, \hat \scale_d$ defined by 
    \begin{align}
    \label{eq:pfd-closed-form-new}
       \hat\scale_i = \frac{\prod_{j=1}^d \hat\decay_i - \decay_j}{\prod_{j \neq i} \hat\decay_i - \hat\decay_j}
    \end{align}
    satisfy the following partial fraction decomposition for the rational function
    \begin{align} \label{eq:pfd-calc-output-scale}
    \hatgenf(x) = \frac{\polyDenom(x)}{\polyInvDenom(x)}
    = \frac{\prod\iind (1 - \decay_i x)}{\prod\iind (1 - \hat \decay_i x)}
    = 1 + \sum\iind \frac{\hat \scale_i x}{1 - \hat\decay_i x} \,.
    \end{align}
    Further, $\hat \scale_1, \ldots, \hat \scale_d$ from \cref{eq:pfd-closed-form-new} are the unique values that satisfy the decomposition of \cref{eq:pfd-calc-output-scale}. They also satisfy the identity
    \begin{align} \label{eq:pdf-identity}
        \sum\iind \frac{\hat \scale_i}{\hat \decay_i} + \frac{\prod\iind \decay_i}{\prod\iind \hat \decay_i} = 1 \,.
    \end{align}
\end{lemma}
\begin{proof}
    \cref{eq:pfd-closed-form-new} can be derived by simplifying Eq. (5.3) of \cite[Lemma 5.2]{dvijotham2024efficient}, but we give a short elementary proof here.
    We start by multiplying \cref{eq:pfd-calc-output-scale} through by $(1 - \hat\decay_i x)$ and taking the limit $x \to \hat \decay_i^{-1}$ to get:
    \[
        \lim_{x \to \hat \decay_i^{-1}} (1 - \hat \decay_i x)\, \hatgenf(x)
        = \lim_{x \to \hat \decay_i^{-1}} \left[(1 - \hat \decay_i x) + \sum_{j=1}^d \frac{\hat \scale_j x \, (1 - \hat \decay_i x)}{1 - \hat \decay_j x} \right] 
        = \frac{\hat \scale_i}{\hat \decay_i} \,.
    \]
    Thus, we can evaluate the scale parameter as
    \begin{align*}
        \hat \scale_i &= \lim_{x \to \hat \decay_i^{-1}} \hat \decay_i (1 - \hat \decay_i x)\, \hatgenf(x)
        = \frac{\hat \decay_i \prod_{j=1}^d (1 - \decay_j / \hat \decay_i)}{\prod_{j \ne i}(1 - \hat \decay_j / \hat \decay_i)} \,,
    \end{align*}
    and rearranging yields \cref{eq:pfd-closed-form-new}. 
    
    Next, we can show \cref{eq:pdf-identity} by taking the limit of $x \to \infty$ in \cref{eq:pfd-calc-output-scale}:
    \[
        \frac{(-1)^d \prod\iind \decay_i}{(-1)^d \prod\iind \hat \decay_i}
        = \lim_{x \to \infty} \frac{\polyDenom(x)}{\polyInvDenom(x)}
        = \lim_{x \to \infty} \left(1 + \sum\iind \frac{\hat \scale_i x}{1 - \hat \decay_i x}\right) = 1 - \sum\iind \frac{\hat \scale_i}{\hat \decay_i} \,.
    \]
    Finally, the uniqueness of the $\scale_i$ parameters follows from the uniqueness properties of the coefficients of a partial fraction decomposition of a rational function $\hatgenf$ whose numerator and denominator are co-prime degree-$d$ polynomials. Specifically, this proof is identical to that of \Cref{prop:partial-frac} with the observation that the degree-$(d-1)$ polynomial $\polyDiff(x) = (\polyInvDenom(x) - \polyDenom(x)) / x$ is co-prime with $\polyInvDenom(x)$, and is omitted for brevity.
\end{proof}

Finally, we argue about the signs of the scale parameters.
\begin{proposition} \label{prop:consts}
    Consider the setting of \Cref{prop:technical-main}. Then, we have that $\hat \scale_i < 0$ for all $i \in [d]$.
\end{proposition}
\begin{proof}
    \Cref{cor:order-of-decays} tells us that $\decay_1 > \hat \decay_1 > \decay_2 > \hat \decay_2 > \cdots > \hat \decay_{d-1} > \decay_d > \hat \decay_d$ (irrespective of the value of $\sum_i \scale_i / \decay_i - 1$).
    We use these relations to argue about the signs of each term in \cref{eq:pfd-closed-form-new}. For example, for $i=1$, we have the term $\hat \decay_1 - \decay_1$ in the numerator is negative but all other terms are strictly positive, leading to $\hat \scale_1 < 0$. In general the expression for $\hat \scale_i$ has 
    $i$ negative terms in the numerator and $i-1$ negative term in the denominator, lead to $\hat \scale_i < 0$.
\end{proof}

\begin{algorithm}[t]
\caption{Inverse BLT Parameterization}
\label{alg:inverse-blt-params}
\begin{algorithmic}[1]
\Require{BLT parameters $\scalev \in \R^d_{++}$ and $\decayv \in (0, 1)^d$ with $\sum\iind \scale_i / \decay_i \neq 1$ and $\scale_i$'s distinct.}
\Ensure{Parameters $\hat\scalev, \hat\decayv \in \R^d$ such that $\BLT_n(\scalev, \decayv)^{-1} = \BLT_n(\hat\scalev, \hat\decayv)$ for all integers $n > 0$.}

\State Define polynomials $\polyDiff, \polyDenom, \polyInvDenom$ as 
\[
\polyDenom(x) = \prod\iind (1 - \decay_i x), \quad
\polyDiff(x) = \sum\iind \scale_i \prod_{j \neq i} (1 - \decay_j x), \quad 
\polyInvDenom(x) = \polyDenom(x) + x\, \polyDiff(x) \,.
\]
\State Find the (real and distinct) roots $\nu_1, \ldots, \nu_d$ of $\polyInvDenom$. One way is to return the eigenvalues of the companion matrix
\[
    \bfM = 
    \left(\begin{array}{ccc|c}
    &\zeros_{d-1}\T && -\polyInvDenom_0 / \polyInvDenom_d \\
    \hline 
    &{ \Large \bfI_{d-1}} &&
    \begin{matrix}
        \,\,-\polyInvDenom_1 /  \polyInvDenom_d \\
        \vdots \\ 
        \,\,- \polyInvDenom_{d-1} / \polyInvDenom_d
    \end{matrix}
    \end{array}  \right)
    \in \R^{d \times d} \,,
\]
where $\polyInvDenom_0, \polyInvDenom_1, \ldots, \polyInvDenom_d$ are the coefficients of $\polyInvDenom$ so that $\polyInvDenom(x) = \sum_{k=0}^d \polyInvDenom_k x^k$,
and $\zeros_m\in\R^m$ denotes the vector of zeros, while $\bfI_{m} \in \R^{m\times m}$ denotes the identity matrix.
\State \label{line:inverse-blt-params}
For $i = 1, \ldots, d$, set the decay parameters $\hat \decay_i = 1/ \nu_i$ and
\[
    \hat \scale_i = \frac{\prod_{j=1}^d \hat\decay_i - \decay_j}{\prod_{j \neq i} \hat\decay_i - \hat\decay_j} \,.
\]
\State \textbf{Return} $\hat\scalev = (\scale_1, \ldots, \scale_d)$ and $\hat\decayv = (\decay_1, \ldots, \decay_d)$.
\end{algorithmic}
\end{algorithm}

\section{Differentiable Algorithms for Inverse BLTs}
\label{sec:algo}

We now give an algorithm to find the parameters corresponding to the inverse BLT.

From the preceding sections proving \Cref{thm:main}, we see that the decay parameters $\hat\decayv$ of the inverse BLT are obtained from the reciprocal of the roots of the polynomial $\polyInvDenom(x) = \polyDenom(x) + x\, \polyDiff(x)$, with scale parameters $\hat\scale_i$ as derived in \Cref{prop:consts-expression}. The entire procedure is summarized in \Cref{alg:inverse-blt-params}.

In particular, the standard procedure to find roots of the polynomial $\polyInvDenom$ is finding the eigenvalues of the companion matrix; numerical procedures to compute the eigenvalues of a matrix are widely available in common software packages. This is based on the following result:

\begin{lemma}[Thm. 3.3.14 of \cite{horn2012matrix}; \cite{kato2013perturbation}]
\label{lem:roots-via-comp-matrix}
    For any $\polyInvDenom_0, \ldots, \polyInvDenom_{d-1} \in \C$, the characteristic polynomial $\polyInvDenom(x) = \mathrm{det}(\bfM - x \bfI_d)$ of the matrix 
    \[
    \bfM = 
    \left(\begin{array}{ccc|c}
    &\zeros_{d-1}\T && -\polyInvDenom_0 \\
    \hline 
    &{ \Large \bfI_{d-1}} &&
    \begin{matrix}
        \,\,-\polyInvDenom_1 \\
        \vdots \\ 
        \,\,- \polyInvDenom_{d-1}
    \end{matrix}
    \end{array}  \right)
    \in \C^{d \times d} \,.
\]
is given by $\polyInvDenom(x) = \polyInvDenom_0 + \polyInvDenom_1 x + \cdots + \polyInvDenom_{d-1} x^{d-1} + x^d$.
In particular, the eigenvalues $\nu_1, \ldots, \nu_d$ (real or complex) of $\bfM$ are the roots of the degree-$d$ polynomial $\polyInvDenom$.
Furthermore, the map $(\polyInvDenom_0, \ldots, \polyInvDenom_{d-1}) \mapsto (\nu_1, \ldots, \nu_d)$ is continuously differentiable over the set
\[
    R =  \Big\{(\polyInvDenom_0, \ldots, \polyInvDenom_{d-1}) \in \C^d \,: \, 
    \text{ the roots } \nu_1, \ldots, \nu_d \text{ of } \polyInvDenom(x) \text { are distinct} \Big\} \,.
\]
\end{lemma}

Thus, we get the following correctness guarantee:

\begin{corollary} \label{corr:alg-correct-diff}
    Given parameters $\scalev \in \R_+^d$ and $\decayv \in (0, 1)^d$ such that $\decay_i's$ are distinct and $\sum\iind \scale_i / \decay_i \neq 1$. Then, \Cref{alg:inverse-blt-params} returns parameters $\hat\scalev, \hat\decayv$ such that $\BLT_n(\scalev, \decayv)^{-1} = \BLT_n(\hat\scalev, \hat\decayv)$ for all $n > 0$.
    Furthermore, the map $(\scalev, \decayv) \mapsto (\hat\scalev, \hat \decayv)$ is continuously differentiable.
\end{corollary}
\begin{proof}
    The correctness of  \Cref{alg:inverse-blt-params} follows from \Cref{thm:main}, \Cref{prop:technical-main,prop:consts}, and \Cref{lem:generating-fn}. 
     We get the following composition of continuously differentiable functions to get from $\scalev, \decayv$ to $\hat\scalev, \hat\decayv$:
    \begin{itemize}
        \item $\scalev, \decayv$ to the coefficients of the polynomials $\polyDiff(x), \polyDenom(x)$, and $\polyInvDenom(x) = \polyDenom(x) + x\, \polyDiff(x)$;
        \item the coefficients of the degree-$d$ polynomial $\polyInvDenom(x)$ to its roots $\nu_1, \ldots, \nu_d$ (which are unique due to \Cref{prop:r_of_mu_i}, and thus this map is continuously differentiable due to \Cref{lem:roots-via-comp-matrix});
        \item $\nu_1, \ldots, \nu_d$ and the coefficients of the polynomial $\polyDenom$ to $\hat\scalev, \hat\decayv$ as per Line~\ref{line:inverse-blt-params} of \Cref{alg:inverse-blt-params}.
    \end{itemize}
    Thus, As a composition of continuously differentiable functions, the map $(\scalev, \decayv) \mapsto (\hat\scalev, \hat\decayv)$ is also continuously differentiable.
\end{proof}

\paragraph{Compatibility with Automatic Differentiation}
The max loss from \cref{eq:maxloss} is also a differentiable function of the BLT parameters $\scalev, \decayv$, thanks to \Cref{corr:alg-correct-diff}.
We can thus optimize $\scalev, \decayv$ to minimize the max loss \eqref{eq:maxloss} using first-order optimization, provided we can find the gradients of the loss w.r.t. $\scalev$ and $\decayv$. This can be achieved with automatic differentiation packages, including JAX and PyTorch which support hardware accelerators like GPUs.
Indeed, the only non-trivial operations (excluding addition, subtraction, multiplication, division) used to obtain the inverse parameters is the eigenvalue computation, and this function is differentiable when the eigenvalues are unique~\cite{kato2013perturbation}.

\section{Discussion and Open Problems}
\label{sec:discussion}

We give an inversion theorem for a family of Buffered Linear Toeplitz (BLT) matrices, a family of parameterized lower-triangular and Toeplitz matrices introduced by \citet{dvijotham2024efficient} for streaming differential privacy with correlated noise. The key contribution is proving that the inverse of a BLT matrix is also a BLT matrix, deriving the parameters of this inverse.  Specifically, we show that under certain conditions on the original BLT parameters, the inverse BLT parameters exhibit desirable properties for differential privacy applications.  Furthermore, we provide a differentiable algorithm for computing the inverse BLT parameters in $O(d^3)$ time, enabling the optimization of BLT mechanisms for private learning and estimation.

There are several interesting open problems in this space. The first one is to find the largest class of (BLT, inverse-BLT) systems that admit an inversion theorem such as \Cref{thm:main}, or an equivalence theorem such as \Cref{thm:calc-output-scale}.
Moreover, we observe theoretically (from the construction of \cite{dvijotham2024efficient}) and empirically that most practically relevant BLTs (in the context of streaming differential privacy) seem to satisfy the constraint $\sum\iind \scale_i / \decay_i < 1$ (or that such a constraint does not hurt). This leads to a practical question: what is the best set of BLTs to optimize over? 

The BLT class as introduced by \citet[Sec 1.2]{dvijotham2024efficient} is more general than the parameterization we give in \cref{eq:blt-param} --- they allow BLTs to be defined by Toeplitz coefficients given by an arbitrary order-$d$ linear recurrence, or equivalently, an arbitrary degree $d$ rational generating function, while the more restricted class we consider only captures rational generating functions with distinct roots and (except in degenerate cases) equal degree in the numerator and denominator. It is an interesting open question whether the generalization to arbitrary linear recurrences yields practical benefit. 

Here, we have some evidence in the affirmative, in that banded lower-triangular Toeplitz matrices are in fact such BLTs with a trivial recurrence (equivalently: a polynomial ordinary generating function). Such matrices have proved useful in DP with multiple participations and/or privacy amplification via randomizaton of the data order \citep{choquette2023amplified,choquette24amplification,mckenna2024scaling}. On the other hand, banded matrices are generally straightforward to reason about without the machinery of general BLTs or rational generating functions, so showing a strict improvement from this generalization remains an interesting open problem.

\bibliographystyle{unsrtnat}
\bibliography{bib}

\appendix

\section{Notation Summary}
\label{sec:notation}

\paragraph{Notation Table}
We summarize our main notation in \Cref{tab:notation}.

\paragraph{Generation Function Summary}
The following lemma summarizes the key results of the generating function that we use throughout.

\begin{lemma}\label{lem:blt-ogfs-new}
We show two characterizations of the generating functions of BLTs. 
\begin{itemize}
    \item For non-zero scale parameters $\scalev \in \R_{++}^d$ and decay parameters $\decayv \in (0, 1)^d$, we define the polynomials  
    \[
    \polyDenom(x) \coloneqq \prod\iind (1 - \decay_i x)
    \qquad \text{and} \qquad
    \polyDiff(x) \coloneqq \sum\iind \frac{\scale_i \polyDenom(x)}{1 - \decay_i x}.
    \]
    Then, there exist $\hat\scalev \in \R^d, \hat\decayv \in \R^d$ such that the statements below hold.
    \item Given non-zero distinct decay parameters $\decayv, \hat\decayv \in (\R \setminus \{0\})^d$, we define the polynomials
    \[
    \polyDenom(x) \coloneqq \prod\iind (1 - \decay_i x) 
    \qquad \text{and} \qquad
    \polyInvDenom(x) \coloneqq \prod\iinD (1 - \hat\decay_i x).
    \]
    Assuming $\decay_i \neq \hat \decay_j$ for all $i, j \in [d]$, there exist $\scalev, \hat\scalev \in \R^d$ such that the statements below hold.
\end{itemize}
Then, in either of the above scenarios, the following equalities hold:
\begin{align*}
  \polyDenom(x) 
    &= \prod\iind (1 - \decay_i x) 
     = \frac{(-1)^d}{M} \prod\iind (x - \mu_i)
    && \text{of degree $d$ with roots $\mu_i = \decay_i^{-1}$ and $M = \prod\iind \mu_i$} \\
  \polyDiff(x) &= \sum\iind \frac{\scale_i \polyDenom(x)}{1 - \decay_i x}  && \text{of degree $d-1$, and}\\
  \polyInvDenom(x) &= \polyDenom(x) + x \polyDiff(x) = \prod\iind (1 - \hat\decay_i x)&& \text{of degree $D \le d$ with roots $\nu_i = \hat\decay\inv_i$ for $i \in [D]$.}
\end{align*}
Further, $\genf(x)$ and $\hatgenf(x)$ are the ordinary generating functions for $\BLT(\scalev, \decayv)$ and $\BLT(\hat\scalev, \hat\decayv)=\BLT(\scalev, \decayv)\inv$:

\begin{alignat}{5}
\label{eq:genfn-c-v2}
\genf(x) 
  &=
  & 1 + x \, \frac{\polyDiff(x)}{\polyDenom(x)}
  &= \frac{\polyDenom(x) + x \polyDiff(x)}{\polyDenom(x)}
  &= \frac{\polyInvDenom(x)}{\polyDenom(x)} 
  & 
  &= 1 + \sum\iind \frac{\scale_i x}{1 - \decay_i x}\,, \\
\label{eq:genfn-cinv-v2}
\hatgenf(x)
  &= \frac{1}{\genf(x)} 
  &= 1 + x \frac{-\polyDiff(x)}{\polyInvDenom(x)}
  &= \frac{\polyDenom(x)}{\polyDenom(x) + x \polyDiff(x)}
  &=\frac{\polyDenom(x)}{\polyInvDenom(x)}
  &
  &= 1 + \sum\iind \frac{\hat\scale_i x}{1 - \hat\decay_i x}.
\end{alignat}
\end{lemma}

\begin{table}[tbhp]
\centering
\renewcommand{\arraystretch}{1.4}%
\begin{tabular}{lp{4.7in}}
\toprule
\textbf{Symbol} & \textbf{Meaning} \\
\midrule
$\strategy[j, k]$ & The $(j,k)$\textsuperscript{th} entry of the matrix $\strategy$. \\
$\bfZ \sim \normalnm{0}{\sigma^2}$ & A random matrix  $\bfZ \in \R^{n \times m}$ whose entries are i.i.d. $\normal(0, \sigma^2)$. \\
$d$ & Number of buffers (positive integer) \\
$[d]$ & The set $\{1, 2, \dots, d\}$. \\
$\scalev \in \R^d$ &  Scale parameters of the BLT. We assume $\scale_i \ge 0$ and $\sum\iind \scale_i < 1$ \\
$\decayv \in \R^d$ &  Decay parameters of the BLT. We assume distinct $\decay_i \in (0, 1)$ for each $i$ \\
$\BLT(\scalev, \decayv)$ & A semi-infinite lower triangular Toeplitz matrix whose first column is given by $1, \sum\iind \scale_i, \sum\iind \scale_i \decay_i, \sum\iind \scale_i \decay_i^2, \sum\iind \scale_i \decay_i^3, \ldots$ \\
$\BLT_n(\scalev, \decayv)$ &  An $n \times n$ lower triangular and Toeplitz matrix which is the principal sub-matrix of $\BLT(\scalev, \decayv)$ \\ 
$\hat \scalev \in \R^d, \hat \decayv^d$ & Scale and decay parameters such that $\BLT(\scalev, \decayv)\inv = \BLT(\hat\scalev, \hat\decayv)$ (whose existence is posited by \Cref{thm:main}) \\
\midrule
$\polyDenom(x)$ & The degree-$d$ polynomial $\prod\iind (1 - \decay_i x)$; $\decay_i$'s are assumed distinct throughout \\ 
$\polyDiff(x)$ & The degree-$(d-1)$ polynomial  $\displaystyle \sum\iind \frac{\scale_i \polyDenom(x)}{1 - \decay_i x }$ \\  
$\polyInvDenom(x)$ & The polynomial $\polyInvDenom(x) = \polyDenom(x) + x \, \polyDiff(x)$; its degree-$D$ can be $d-1$ or $d$ \\
$D$ & Degree of the polynomial $\polyInvDenom$ \\
$\nu_1, \ldots, \nu_D$ & Roots of $\polyInvDenom$; the decay parameter $\hat \decayv$ of the inverse BLT satisfies $\hat\decay_i = \nu_i^{-1}$ for $i \in [D]$\\
$\genf$ & 
\begin{tabular}{l}
Generating function of the first column of $\BLT(\scalev, \decayv)$. It satisfies \\
\multicolumn{1}{l}{
$\displaystyle \genf(x) = 1 + x \frac{\polyDiff(x)}{\polyDenom(x)} = \frac{\polyInvDenom(x)}{\polyDenom(x)} = 1 + \sum\iind \frac{\scale_i x}{1 - \decay_i x}$}
\end{tabular} \\
$\hatgenf$ & 
\begin{tabular}{l}
Generating function of the first column of $\BLT(\scalev, \decayv)\inv = \BLT(\hat \scalev, \hat \decayv)$. It satisfies \\
\multicolumn{1}{l}{
$\displaystyle \hatgenf(x) = 1 + x \frac{-\polyDiff(x)}{\polyInvDenom(x)} = \frac{\polyDenom(x)}{\polyInvDenom(x)} = 1 + \sum\iind \frac{\hat\scale_i x}{1 - \hat\decay_i x}$}
\end{tabular} \\
\midrule
$\mu_1, \ldots, \mu_d$ & Roots of the polynomial $\polyDenom(x)$; satisfies $\mu_i = \decay_i\inv$ and sorted in ascending order \\
$M$ & Shorthand for $\prod\iind \mu_i$\\
$\beta_1,\ldots, \beta_d$ & Constants that satisfy $\beta_i = \polyDiff(\mu_i)$  \\
\bottomrule
\end{tabular}
\caption{Summary of main notation. Matrices and vectors are denoted in boldface.}\label{tab:notation}
\end{table}

\section{More Illustrations and Details}

\begin{figure}[p]
    \centering
    \adjustbox{max height=0.75\textheight}{
    \includegraphics[width=0.99\linewidth]{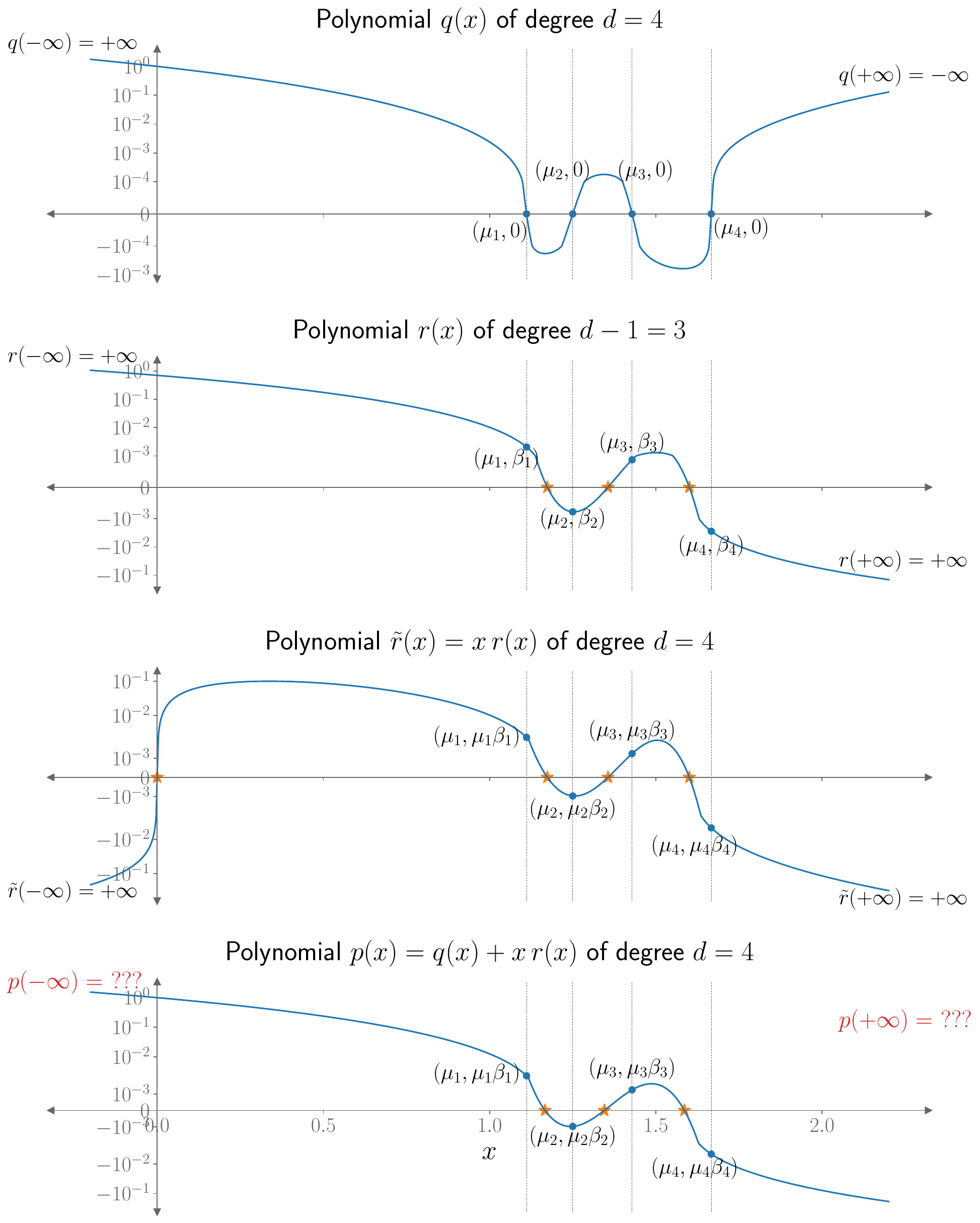}}
    \caption{\small {Illustrations of the polynomials $\polyDiff, \polyDenom, \polyInvDenom$ for $d=4$ in symmetrical log scale}. This is the counterpart of \Cref{fig:deg-5} for even degree $d$.
    }
    \label{fig:deg-4}
\end{figure}

\begin{figure}[t]
    \centering
    \includegraphics[width=0.7\linewidth]{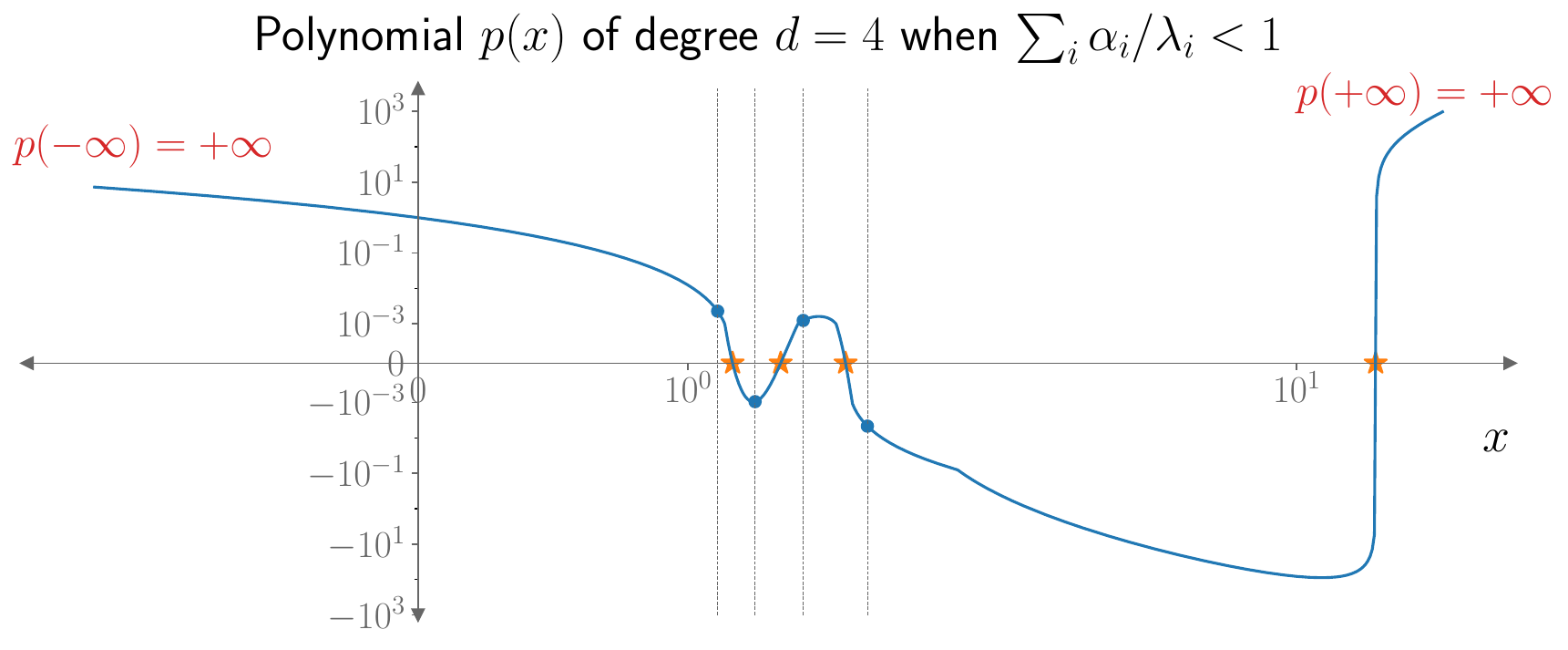}
    \includegraphics[width=0.7\linewidth]{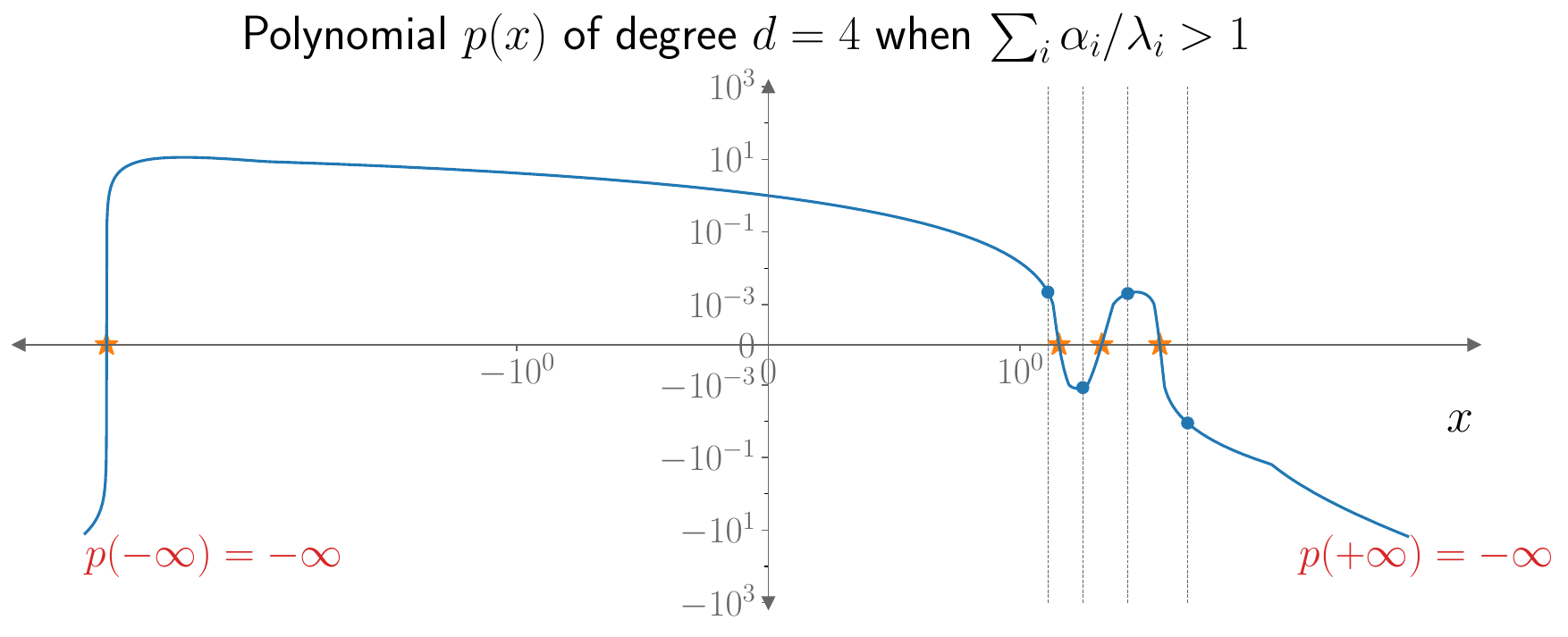}
    \caption{\small The counterpart of \Cref{fig:deg-5b} for even degree: this plot shows examples for $d=4$ and is continued from \Cref{fig:deg-4}, which shows $d-1$ roots of $\polyInvDenom(x) = \polyDenom(x) + x\,\polyDiff(x)$. This figure illustrates how the final $d$\textsuperscript{th} root of $\polyInvDenom(x)$ depends on the BLT parameters $\scalev, \decayv$. As previously, the dotted lines denote the roots $\mu_1, \ldots, \mu_d$ of $\polyDenom(x)$ (where $\mu_i = 1/\decay_i$) and the orange star denotes the roots of $\polyInvDenom(x)$.
    }
    \label{fig:deg-4b}
\end{figure}

We give examples illustrating the behaviors of the polynomials $\polyDiff, \polyDenom, \polyInvDenom$ defined in \Cref{lem:blt-ogfs} for degree $d=5$ in \Cref{fig:deg-5,fig:deg-5b} and for degree $d=4$ in \Cref{fig:deg-4,fig:deg-4b}. They use the following BLT parameters:
\begin{itemize}
    \item \Cref{fig:deg-5} and the top row of \Cref{fig:deg-5b}: $\scalev=(0.2, 0.15, 0.1, 0.1, 0.1)$ and $\decayv = (0.9, 0.8, 0.7, 0.6, 0.5)$. We have $\sum\iind \scale_i / \decay_i \approx 0.919 < 1$.
    \item Bottom row of \Cref{fig:deg-5b}: $\scalev = (0.2, 0.15, 0.2, 0.2, 0.2)$ and $\decayv = (0.9, 0.8, 0.7, 0.6, 0.5)$.
    We have $\sum\iind \scale_i = 0.95 < 1$ and $\sum\iind \scale_i / \decay_i \approx 1.43 > 1$.
    \item \Cref{fig:deg-4} and the top row of \Cref{fig:deg-4b}: $\scalev=(0.25, 0.2, 0.15, 0.1)$ and $\decayv=(0.9, 0.8, 0.7, 0.6)$. We have $\sum\iind \scale_i / \decay_i \approx 0.909 < 1$.
    \item Bottom row of \Cref{fig:deg-4b}: $\scalev=(0.24, 0.24, 0.24, 0.24)$ and $\decayv=(0.9, 0.8, 0.7, 0.6)$.
    We have $\sum\iind \scale_i = 0.96 < 1$ and $\sum\iind \scale_i / \decay_i \approx 1.31 > 1$.
\end{itemize}

\end{document}